%% file: ted_faster.tex
\documentclass[a4paper,USenglish,cleveref,autoref,thm-restate,numberwithinsect]{lipics-v2021}

\title{Faster Algorithm for Bounded Tree Edit Distance in the Low-Distance Regime}

\author{Tomasz Kociumaka}{Max Planck Institute for Informatics, Saarland Informatics
Campus, Saarbrücken, Germany}{tomasz.kociumaka@mpi-inf.mpg.de}{https://orcid.org/0000-0002-2477-1702}{}

\author{Ali Shahali}{Sharif University of Technology, Tehran, Iran}{alishahali1382@gmail.com}{https://orcid.org/0009-0008-6181-8881}{The work of was carried out mostly during a summer internship at the Max Planck Institute for Informatics.}

\authorrunning{T. Kociumaka and A. Shahali} 

\Copyright{Tomasz Kociumaka and Ali Shahali} 



\category{} 

\relatedversion{} 

\nolinenumbers 

\hideLIPIcs

\EventEditors{John Q. Open and Joan R. Access}
\EventNoEds{2}
\EventLongTitle{42nd Conference on Very Important Topics (CVIT 2016)}
\EventShortTitle{CVIT 2016}
\EventAcronym{CVIT}
\EventYear{2016}
\EventDate{December 24--27, 2016}
\EventLocation{Little Whinging, United Kingdom}
\EventLogo{}
\SeriesVolume{42}
\ArticleNo{23}

\usepackage{amsmath}
\usepackage{amsthm}
\usepackage{amssymb}
\usepackage{thmtools}
\usepackage{mathtools}
\usepackage{thm-restate}
\usepackage{enumitem}
\usepackage{dsfont}
\usepackage{xspace}
\usepackage{comment}
\usepackage{xparse}
\usepackage[ruled,noend,resetcount,linesnumbered]{algorithm2e}
\usepackage{tikz}
\usepackage[draft]{fixme}

\fxsetup{mode=multiuser,theme=color, layout=inline}
\FXRegisterAuthor{tk}{atk}{\color{blue} {\bf Tomasz}}
\FXRegisterAuthor{as}{aas}{\color{red} {\bf Ali}}
\FXRegisterAuthor{rvw}{arvw}{\color{green} {\bf Reviewer}}
\SetKwProg{Fn}{}{:}{}

\newtheorem{fact}[theorem]{Fact}

\newcommand{\op}{\textrm{\texttt{(}}}
\newcommand{\cl}{\textrm{\texttt{)}}}

\newcommand{\B}{\mathcal{B}}

\newcommand{\sub}{\subseteq}

\newcommand{\Oh}{\mathcal{O}}
\newcommand{\Ohtilde}{\widetilde{\Oh}}
\newcommand{\ed}{\mathsf{ed}}
\newcommand{\ted}{\mathsf{ted}}

\newcommand{\per}{\mathsf{per}}

\newcommand{\emptystring}{\varepsilon}

\newcommand{\poly}{\mathrm{poly}}

\newcommand{\A}{\mathcal{A}}
\newcommand{\M}{\mathcal{M}}

\newcommand{\dd}{\mathinner{.\,.\allowbreak}}

\newcommand{\Zz}{\mathbb{Z}_{\ge 0}}
\newcommand{\Zp}{\mathbb{Z}_{+}}

\newcommand{\cost}{\mathsf{cost}}

\newcommand{\Qf}{\mathcal{Q}}

\newcommand{\AGW}{\mathsf{AG}^w}
\newcommand{\AGw}{\AGW}

\newcommand{\fragmentco}[2]{[#1\dd #2)}

\newcommand{\fragmentcc}[2]{[#1\dd #2]}
\newcommand{\position}[1]{[#1]}
\newcommand{\w}[2]{\ensuremath w(#1,#2)}
\newcommand{\edwa}[3]{\mathsf{ed}^w_{#1}(#2,#3)}
\newcommand{\Als}{\mathbf{A}} 
\renewcommand{\aa}{\Als}
\newcommand{\ta}{\mathbf{FA}}
\newcommand{\bta}[1]{\mathbf{BFA}_{#1}}
\newcommand{\lbl}{\mathsf{lbl}}
\newcommand{\wtype}{\eSigma^2 \to \mathbb{R}_{\ge 0}}

\newcommand{\Pc}{\mathcal{P}}
\newcommand{\D}{\mathcal{D}}
\newcommand{\HC}{\mathcal{H}}
\newcommand{\DC}{\mathcal{D}}

\newcommand{\BF}{\mathbf{B}_k(F)}
\newcommand{\Bk}[1]{\mathbf{B}_k(#1)}
\newcommand{\BP}{\mathbf{B}_k(P)}
\newcommand{\I}{\mathcal{I}}
\newcommand{\black}[1]{\mathsf{black}_k(#1)}
\newcommand{\red}[1]{\mathsf{red}_k(#1)}
\newcommand{\context}[2]{\langle #1; #2 \rangle}

\makeatletter
\patchcmd{\algocf@makecaption@ruled}{\hsize}{\textwidth}{}{} 
\patchcmd{\@algocf@start}{0em}{0em}{}{} 
\makeatother

\def\twoheadleadsto{\tikz[baseline=(a.base)]{%
        \node at (.2,0) {\(\leadsto\)};%
        \fill[white] (-.1,-.1+.012) -- (.25,-.1+.012) -- (.25,.113) --
        (-.1,.113) --cycle;
        \node at (.125,0) {\(\leadsto\)};%
        \node (a) at (.4/2,-.0) {\phantom{\(\leadsto\)}};%
}}

\newcommand{\onto}{\twoheadleadsto}
\newcommand{\ponto}[1]{\twoheadleadsto\kern-.3em{}_{#1}\kern.3em}

\newcommand{\ated}{\widetilde{\smash{\mathsf{ted}}}}
\newcommand{\bted}[1]{\mathsf{bted}_{#1}}

\date{}

\begin{document}

\maketitle

\begin{abstract}
        \input{src/abstract.tex}
\end{abstract}

\section{Introduction}
\input{src/intro.tex}

\section{Preliminaries}
\input{src/prelim.tex}

\section{\texorpdfstring{\boldmath $\Oh(n k^2 \log n)$}{O(nk² log n)}-Time Algorithm}\label{sec:our_aj}
\input{src/klein_overview.tex}

\subsection{Our Algorithm}
\input{src/new_nk2_overview.tex}

\section{Faster Algorithm for Repetitive Inputs}\label{sec:optimization}
\input{src/matrix_algorithm_overview.tex}

\section{Universal Kernel with Improved Repetitiveness Guarantees}\label{sec:our_kernel}

\input{src/new_kernel_overview.tex}

\section{Summary}\label{sec:conclusions}
\input{src/conclusions.tex}

\bibliography{ted}

\appendix

\section{\texorpdfstring{\boldmath $\Oh(n k^2 \log n)$}{O(nk² log n)}-Time Algorithm: Details}\label{app:our_aj}
\input{src/klein.tex}

\subsection{Our Algorithm}
\input{src/new_nk2.tex}
\section{Faster Algorithm for Repetitive Inputs: Details}\label{app:optimization}
\input{src/matrix_algorithm.tex}

\section{Universal Kernel with Improved Periodicity Guarantees: Details}\label{app:our_kernel}

\input{src/new_kernel.tex}

\end{document}

%% file: src/abstract.tex
The \emph{tree edit distance} is a natural dissimilarity measure between rooted ordered trees whose nodes are labeled over an alphabet $\Sigma$. 
It is defined as the minimum number of node edits---insertions, deletions, and relabelings---required to transform one tree into the other. 
The weighted variant assigns costs $\geq 1$ to edits (based on node labels), minimizing total cost rather than edit count.

The unweighted tree edit distance between two trees of total size $n$ can be computed in $\Oh(n^{2.6857})$ time; in contrast, determining the weighted tree edit distance is fine-grained equivalent to the All-Pairs Shortest Paths (APSP) problem and requires $n^3 / 2^{\Omega(\sqrt{\log n})}$ time [Nogler, Polak, Saha, Vassilevska Williams, Xu, Ye; STOC'25].
These impractical super-quadratic times for large, similar trees motivate the bounded version, parameterizing runtime by the distance~$k$ to enable faster algorithms for $k \ll n$.

Prior algorithms for bounded unweighted edit distance achieve \(\Oh(nk^2\log n)\) [Akmal \& Jin; ICALP’21] and \(\Oh(n + k^7\log k)\) [Das, Gilbert, Hajiaghayi, Kociumaka, Saha; STOC'23].
For weighted, only \(\Oh(n + k^{15})\) is known [Das, Gilbert, Hajiaghayi, Kociumaka, Saha; STOC'23].  

We present an $\Oh(n + k^6 \log k)$-time algorithm for bounded tree edit distance in both weighted/unweighted settings.
First, we devise a simpler weighted $\Oh(nk^2 \log n)$-time algorithm.
Next, we exploit periodic structures in input trees via an optimized universal kernel: modifying prior $\Oh(n)$-time $\Oh(k^5)$-size kernels to generate such structured instances, enabling efficient analysis.  

%% file: src/intro.tex
The edit distance between two strings---the minimum cost of insertions, deletions, and substitutions needed to transform one string into the other---is one of the most widely used string similarity measures with numerous algorithmic applications. 
It provides a robust model for comparing sequential data and underpins techniques in fields such as computational biology, natural language processing, and text correction.
Nevertheless, many data types exhibit hierarchical rather than linear structure. For example, RNA secondary structures, syntactic parse trees of natural language sentences, and hierarchical representations of documents and code all store information in tree-like forms.
One can linearize such data and still use string edit distance, but the resulting alignments disregard the original hierarchical structure, and thus more expressive similarity measures are needed in most scenarios. 

First introduced by Selkow~\cite{S77}, tree edit distance generalizes the notion of string edit distance to rooted ordered trees and forests with nodes labeled over an alphabet $\Sigma$.
It quantifies the dissimilarity between two forests as the minimal cost of a sequence of node edits---insertions, deletions, and relabelings---required to transform one forest into the other. 
This distance naturally captures both structural and label-based differences and serves as a core primitive in various algorithmic and applied domains, including computational biology \cite{G97,SZ90,HTGK03,W95}, analysis of structured data (such as XML and JSON files)~\cite{BGK03,CAM02,Cha99,FLMM09,WDC03}, image processing \cite{BK99,KTSK00,KSK01,SKK04}, and information extraction~\cite{RGSF04,YDCC13}; see~\cite{B05,A10} for surveys.

Tai~\cite{T79} proposed the first polynomial-time algorithm for computing the tree edit distance: a dynamic programming procedure running in $\Oh(n^6)$ time, where $n$ is the total number of nodes in the input forests. 
In the following decades, a sequence of works progressively improved the runtime. 
Zhang and Shasha~\cite{ZS89} introduced an $\Oh(n^4)$-time algorithm, and then Klein~\cite{K98} brought the complexity down to $\Oh(n^3 \log n)$.
Subsequently, Demaine, Mozes, Rossman, and Weimann~\cite{DMRW10} presented an $\Oh(n^3)$-time solution, whereas Bringmann, Gawrychowski, Mozes, and Weimann~\cite{BGMW20} proved that a hypothetical $n^{3-\Omega(1)}$-time algorithm would violate the All-Pairs Shortest Paths (APSP) hypothesis in fine-grained complexity.
Very recently, Nogler, Polak, Saha, Vassilevska Williams, Xu, and Ye~\cite{NPSVXY25} showed that computing the tree edit distance is, in fact, equivalent to the APSP problem, and the $n^3/2^{\Omega(\sqrt{\log n})}$ running time can be inherited from the state of the art for the latter task~\cite{W18}.

The aforementioned conditional lower bounds apply to the weighted tree edit distance only, where the cost of each edit depends on the involved labels.
In the unweighted version with unit costs, a series of recent results achieved truly subcubic running time using fast matrix multiplication.
Mao~\cite{M21} presented an $\Oh(n^{2.9546})$-time solution, which was subsequently improved by D\"urr~\cite{D23} to $\Oh(n^{2.9148})$ and by Nogler et al.~\cite{NPSVXY25} to $\Ohtilde(n^{(3+\omega)/2})=\Oh(n^{2.6857})$, where $\omega$ is the fast matrix multiplication exponent and $\Ohtilde(\cdot)$ hides $\poly \log n$ factors.
The state-of-the-art conditional lower bound, inherited from string edit distance~\cite{BI18} and assuming the Orthogonal Vectors Hypothesis, prohibits $n^{2-\Omega(1)}$-time algorithms already if $|\Sigma|=1$.

\subparagraph*{Bounded Tree Edit Distance}
Even though the decades of research resulted in significantly faster tree edit distance algorithms, they also revealed fundamental challenges that explain the difficulty of this problem.
A natural way of circumventing these barriers, paved by classical work for string edit distance~\cite{Ukk85,Mye86,LV88}, is to parameterize the running time not only in terms of the length $n$ of the input forests but also the value $k$ of the computed distance.
This version of tree edit distance has originally been studied in the unweighted setting only. 
Touzet~\cite{T05} adapted the ideas of Zhang and Shasha~\cite{ZS89} to derive an $\Oh(nk^3)$-time algorithm, whereas Akmal and Jin~\cite{AJ21} improved the running time to $\Oh(nk^2 \log n)$ based on the approach of Klein~\cite{K98}.
The latter running time remains the state-of-the-art for medium distances.
In the low-distance regime, Das, Gilbert, Hajiaghayi, Kociumaka, Saha, and Saleh~\cite{DGHKSS22} achieved an $\Ohtilde(n+k^{15})$-time solution and subsequently improved the running time to $\Oh(n+k^{7}\log k)$~\cite{DGHKS23}.
They also studied the weighted version of the problem and presented an $\Oh(n+k^{15})$-time algorithm assuming that the weight function is normalized (each edit costs at least one unit; this prohibits arbitrary scaling) and satisfies the triangle inequality.

Still, the $\Ohtilde(n+\poly(k))$ running times for tree edit distance are much larger than for string edit distance and the related Dyck edit distance problem, which asks for the minimum cost of edits needed to make a given string of parentheses balanced (so that it represents a node-labeled forest).
In the unweighted setting, the state-of-the-art running times are $\Oh(n+k^2)$ for string edit distance~\cite{LV88} and $\Oh(n+k^{5.442})$ for Dyck edit distance~\cite{FGKKPS24}.
In the presence of weights, these complexities increase to $\Ohtilde(n+\sqrt{nk^3})\le \Ohtilde(n+k^3)$~\cite{CKW23} and $\Oh(n+k^{12})$~\cite{DGHKS23}, respectively. 
Thus, tree edit distance is a natural candidate for improvements.

\subparagraph*{Our Results} 
As the main contribution, we bring the time complexity of tree edit distance to $\Ohtilde(n+k^6)$.
Even for unit weights, this improves upon the state of the art if $n^{1/7} \ll k \ll n^{1/4}$. 

\begin{restatable}{theorem}{thmmain}\label{thm:main}
There exists a deterministic algorithm that, given two forests $F$ and $G$ with $n$ nodes in total, each with a label from an alphabet $\Sigma$, and oracle access to a normalized weight function $w:(\Sigma\cup\{\emptystring\})^2 \to \mathbb{R}_{\ge 0}$ satisfying the triangle inequality, determines the tree edit distance $k\coloneqq \ted^w(F,G)$ in $\Oh(n+k^6 \log k)$ time.
\end{restatable}

We also prove that the $\Oh(nk^2 \log n)$ running time by Akmal and Jin~\cite{AJ21} remains valid in the presence of weights. 
This approach gives the best complexity when $n^{1/4} \ll k \ll n$.
\begin{restatable}{theorem}{thmaj}\label{thm:aj}
There exists a deterministic algorithm that, given two forests $F$ and $G$, with $n$ nodes in total, each with a label from an alphabet $\Sigma$, and oracle access to a normalized weight function $w:(\Sigma\cup\{\emptystring\})^2 \to \mathbb{R}_{\ge 0}$, determines $k\coloneqq \ted^w(F,G)$ in $\Oh(nk^2 \log n)$ time.
\end{restatable}

We believe that the original approach of Akmal and Jin~\cite{AJ21} supports the weighted setting with minimal adjustments; nevertheless, we prove \cref{thm:aj} using a slightly different algorithm that is better-suited for further optimizations.
Both solutions are variants of Klein's dynamic programming~\cite{K98} with some states pruned.
The main difference is that we view the input forests as balanced sequences of parentheses and cast the tree edit distance as the minimum cost of a string alignment satisfying a certain consistency property.
In our interpretation (\cref{sec:our_aj}), it is almost trivial to see that the classic pruning rules, originally devised to compute string edit distance in $\Oh(nk)$ time~\cite{Ukk85,Mye86}, can be reused for tree edit distance in $\Oh(nk^2\log n)$ time.
In contrast, Akmal and Jin~\cite[Lemma 9]{AJ21} spend over four pages (on top of the analysis of~\cite{K98}) to bound the number of states surviving their pruning rules.

\cref{thm:aj} combined with the results of~\cite{DGHKS23} lets us derive an $\Oh(n+k^7\log k)$-time algorithm for weighted tree edit distance, matching the previously best running time for unit weights.
This is due to the \emph{universal kernel} that transforms the input forests to equivalent forests of size $\Oh(k^5)$ preserving the following capped tree edit distance value:
\[\ted^w_{\le k}(F,G) = \begin{cases}
    \ted^w(F,G) & \text{if }\ted^w(F,G)\le k,\\
    \infty & \text{otherwise.}
\end{cases}
\]

\begin{theorem}[{\cite[Corollary 3.20]{DGHKS23}}]\label{thm:prev_kernel}
There exists a linear-time algorithm that, given forests $F,G$ and an integer $k\in \Zp$, constructs forests $F',G'$ of size $\Oh(k^5)$ such that $\ted^w_{\le k}(F,G) = \ted^w_{\le k}(F',G')$ holds for every normalized weight function $w$ satisfying the triangle inequality.
\end{theorem}

To bring the time complexity from $\Ohtilde(n+k^7)$ to $\Ohtilde(n+k^6)$, we adapt both~\cref{thm:aj,thm:prev_kernel}.
The main novel insight is that the dynamic-programming procedure behind~\cref{thm:aj} behaves predictably while processing regions with certain repetitive (periodic) structure.
Upon an appropriate relaxation of the invariant that the dynamic-programming values satisfy, processing each repetition of the period can be interpreted as a single min-plus matrix-vector multiplication.
When the period repeats many times, the same matrix is used each time, and thus we can raise the matrix to an appropriate power (with fast exponentiation) and then compute a single matrix-vector product; see \cref{sec:optimization} for details.
This can be seen as a natural counterpart of an optimization used in~\cite{CKW23} for string edit distance.%
\footnote{The matrices arising in the string edit distance computation in~\cite{CKW23} are $\Oh(k)\times \Oh(k)$ Monge matrices, which allows for computing each matrix-vector product in $\Oh(k)$ time and each matrix-matrix product in $\Oh(k^2)$ time. 
In the context of tree edit distance, the Monge property is no longer satisfied, so the complexities increase to $\Oh(k^2)$ and $\Oh(k^3)$ respectively; in fact, computing the (weighted) tree edit distance is, in general, as hard as computing the min-plus product of two $n\times n$ arbitrary matrices~\cite{NPSVXY25}.}

Intuitively, the worst-case instances for the kernelization algorithm of~\cref{thm:prev_kernel} are close to having the necessary structure for our optimization to improve the worst-case running time.
Unfortunately, the implementation in~\cite{DGHKS23} controls the forest sizes only, and it needs to be modified to keep track of the more subtle instance difficulty measure.
The biggest challenge is that the original algorithm proceeds in logarithmically many steps, each shrinking the forests by a constant fraction.
In every step, the forests are partitioned into $\Oh(k)$ pieces of size $\Oh(n/k)$, and a constant fraction of pieces is shrunk to size $\Oh(k^4)$ each.
A partition into pieces of equal ``difficulty'' would be much more challenging, so we instead take a different approach that lets us implement the reduction in a single step; see \cref{sec:our_kernel} for details.
Notably, the kernelization algorithms for weighted string and Dyck edit distance already have single-step implementations in~\cite{DGHKS23}, but they crucially rely on the fact that the unit-weight variants of these problems can be solved faster, unlike for tree edit distance.

\paragraph*{Related Work}
The classical definition of tree edit distance~\cite{S77} involves labeled rooted trees (or forests) with node labels.
Many other variants have also been studied allowing, among others, unordered~\cite{YHK14,SM20}, unrooted~\cite{SM20}, and edge-labeled~\cite{AFT10} trees; see also the surveys~\cite{A10,B05}.

The tree edit distance problem becomes easier not only when the computed distance is small but also when the input forests are of small depth~\cite{ZS89,CGMW22} or when the approximate distance suffices~\cite{AFT10,BGHS19,SS22}.
There is also a body of work focusing on practical performance and empirical evaluation, including sequential~\cite{PA15,PA16} and parallel~\cite{FLZ24} implementations.

%% file: src/prelim.tex
\newcommand{\PSigma}{\mathsf{P}_\Sigma}
\newcommand{\FSigma}{\mathcal{F}_\Sigma}
\newcommand{\FsSigma}{\mathcal{F}^{\mathsf{sub}}_\Sigma}
\newcommand{\eSigma}{\bar{\Sigma}}
\newcommand{\node}[2]{\mathsf{node}_{#1}(#2)}
\newcommand{\match}[2]{\mathsf{mate}_{#1}(#2)}

Consistently with~\cite{DGHKS23}, we identify forests with the underlying Euler tours, interpreted as balanced strings of parentheses. 
This representation is at the heart of space-efficient data structures on trees~\cite{MR01,NS14}, and it allows for a simple definition of tree edits and a seamless transfer of many tools from strings to forests.
Notably, Klein's algorithm~\cite{K98} already uses substrings of the Euler tours of the input forests to identify the dynamic-programming states.

For an alphabet $\Sigma$, let $\PSigma := \bigcup_{a\in \Sigma}\{\op_a,\cl_a\}$ denote the set parentheses with labels in~$\Sigma$.
As formalized next, a \emph{forest} with node labels over $\Sigma$ is a \emph{balanced} string of parentheses in~$\PSigma$. 

\begin{definition}\label{def:forest}
    The set of \emph{forests} with labels over $\Sigma$ (\emph{balanced strings} over~$\PSigma$) is the smallest subset $\FSigma\sub \PSigma^*$ satisfying the following conditions:
\begin{itemize}
    \item $\emptystring \in \FSigma$,
    \item $F\cdot G \in \FSigma$ for every $F,G\in \FSigma$,
    \item $\op_a \cdot F \cdot \cl_a \in \FSigma$ for every $F\in \FSigma$ and $a\in \Sigma$.
\end{itemize}
\end{definition}

For a forest $F$, we define the set of \emph{nodes} $V_F$ as the set of intervals $[i\dd j]\subseteq [0\dd |F|)$ such that $F[i]$ is an opening parenthesis, $F[j]$ is a closing parenthesis, and $F(i\dd j)$ is balanced.
A forest $F$ is a \emph{tree} if $[0\dd |F|)\in V_F$.
For a node $u= [i\dd j]\in V_F$, we denote the positions of the opening and the closing parenthesis by $o(u):=i$ and $c(u):= j$.
The label of $u$, that is, the character $a\in \Sigma$ such that $F[o(u)]=\op_a$ and $F[c(u)]=\cl_a$, is denoted by $\lbl(u)$.
A simple inductive argument shows that $V_F$ forms a laminar family and, for every $i\in [0\dd |F|)$, there is a unique node $u\in V_F$ such that $i\in \{o(u),c(u)\}$; we denote this node by $\node{F}{i}$.
We also write $\match{F}{i}$ for the paired parenthesis, that is, $\match{F}{i}=j$ if $\{i,j\}=\{o(u),c(u)\}$.

We say that a node $u\in V_F$ \emph{is contained} in a fragment $F[i\dd j)$ of a forest $F$ 
if $i \le o(u) < c(u)< j$; we denote by $V_{F[i\dd j)}\subseteq V_F$ the set of nodes contained in $F[i\dd j)$.
Moreover, $u$ \emph{enters} $F[i\dd j)$ if $o(u) < i \le c(u) < j$
and $u$ \emph{exits} $F[i\dd j)$ if $i \le o(u) < j \le c(u)$.
In either of these two cases, we also say that $u$ \emph{straddles} $F[i\dd j)$.

We denote by $F_{[i\dd j)}$ the \emph{subforest of $F$ induced by $F[i\dd j)$}, which we obtain from $F$ by deleting (the parentheses corresponding to) all nodes except for those contained in $F[i\dd j)$.
Alternatively, one can obtain $F_{[i\dd j)}$ from $F[i\dd j)$ by deleting the opening parenthesis of every node that exits $F[i\dd j)$ and the closing parenthesis of every node that enters $F[i\dd j)$.


\subsection{Tree Edits, Forest Alignments, and Tree Edit Distance}

For an alphabet $\Sigma$, we define $\eSigma := \Sigma\cup\{\emptystring\}$, where $\emptystring$ is the empty string.
We say that a function $w : \wtype$ is \emph{normalized} if $w(a,a)=0$ and $w(a,b)\ge 1$ hold for distinct $a,b\in \eSigma$. 

Tree edit distance is classically defined using elementary edits transforming $F\in \FSigma$:
\begin{description}
    \item[Node insertion] produces $F[0\dd i) \cdot \op_a \cdot F[i\dd j) \cdot \cl_a \cdot F[j\dd |F|)$ for a balanced fragment $F[i\dd j)$ and a label $a\in \Sigma$, at cost $w(\emptystring, a)$.
    \item[Node relabeling] produces $F[0\dd o(u)) \cdot \op_a \cdot F(o(u)\dd c(u)) \cdot \cl_a \cdot F(c(u)\dd |F|)$ for a node $u\in V_F$ and a label $a\in \Sigma$, at cost $w(\lbl(u),a)$.
    \item[Node deletion] produces $F[0\dd o(u)) \cdot F(o(u)\dd c(u)) \cdot F(c(u)\dd |F|)$ for a node $u\in V_F$, at cost $w(\lbl(u),\emptystring)$.
\end{description}
The tree edit distance $\ted^w(F,G)$ of two forests $F,G\in\FSigma$ is then defined as the minimum cost of a sequence of edits transforming $F$ to $G$.
In this context, without loss of generality, we can replace $w$ by its transitive closure. 
If, say $w(a,\emptystring) > w(a,b)+w(b,\emptystring)$, instead of directly deleting a node with label $a$, it is more beneficial to first change its label to $b$ and only then perform the deletion.
When $w$ satisfies the triangle inequality, we are guaranteed that an inserted or relabeled node is never modified (deleted or relabeled) again.
Consistently with modern literature~\cite{BGMW20,DGHKS23,NPSVXY25}, we use a more general \emph{alignment-based} definition of $\ted^w(F,G)$ that enforces the latter condition even if $w$ does not necessarily satisfy the triangle inequality.

\begin{definition}[Alignment Graph~{\cite{CKW23}}]\label{def:alignment-graph}
    For strings $X, Y\in \Sigma^*$ and a weight function $w: \wtype$,
    we define the \emph{alignment graph} $\AGw(X, Y)$ as a grid graph with vertices
    $\fragmentcc{0}{|X|}\times \fragmentcc{0}{|Y|}$ and the following directed edges:
    \begin{itemize}
        \item horizontal edges $(x,y)\to (x+1,y)$ of cost $\w{X\position{x}}{\emptystring}$
            for $(x,y)\in \fragmentco{0}{|X|}\times \fragmentcc{0}{|Y|}$,
        \item vertical edges $(x,y)\to (x,y+1)$ of cost $\w{\emptystring}{Y\position{y}}$
            for $(x,y)\in \fragmentcc{0}{|X|}\times \fragmentco{0}{|Y|}$, and
        \item diagonal edges $(x,y)\to (x+1,y+1)$ of cost $\w{X\position{x}}{Y\position{y}}$
            for $(x,y)\in \fragmentco{0}{|X|}\times \fragmentco{0}{|Y|}$.
            \qedhere
    \end{itemize}
\end{definition}

The alignment graph allows for a concise definition of a string \emph{alignment}.

\begin{definition}[Alignment]
    For strings $X, Y\in \Sigma^*$ and a weight function $w: \wtype$,
    an \emph{alignment} of $X\fragmentco{x}{x'}$ onto
    $Y\fragmentco{y}{y'}$, denoted by $\A: X\fragmentco{x}{x'} \onto Y\fragmentco{y}{y'}$,
    is a path from \((x,y)\) to \((x',y')\) in \(\AGw(X, Y)\), interpreted as a sequence of vertices.
    The \emph{cost} $\edwa{\A}{X\fragmentco{x}{x'}}{Y\fragmentco{y}{y'}}$ of the alignment
    \(\A\) is the total costs of the edges that belong to $\A$.

    We write
    $\Als(X\fragmentco{x}{x'}, Y\fragmentco{y}{y'})$
    for the set of all alignments of $X\fragmentco{x}{x'}$ onto $Y\fragmentco{y}{y'}$.
\end{definition}

\newcommand{\hx}{\hat{x}}
\newcommand{\hy}{\hat{y}}
The edges of an alignment $\A = \Als(X\fragmentco{x}{x'}, Y\fragmentco{y}{y'})$ can be interpreted as follows:
\begin{itemize}
    \item If $\A$ includes an edge $(\hx,\hy)\to (\hx+1,\hy)$ for some $\hx\in \fragmentco{x}{x'}$ and $\hy\in \fragmentcc{y}{y'}$, then $\A$ \emph{deletes} $X\position{\hx}$,  denoted by \(X\position{\hx} \ponto{\A} \varepsilon\).
    \item If $\A$ includes an edge $(\hx,\hy)\to (\hx,\hy+1)$ for some $\hx\in \fragmentcc{x}{x'}$ and $\hy\in \fragmentco{y}{y'}$, then  $\A$ \emph{inserts} $Y\position{\hy}$,  denoted by \(\varepsilon \ponto{\A} Y\position{\hy}\).
    \item If $\A$ includes an edge $(\hx,\hy)\to (\hx+1,\hy+1)$ for some $\hx\in \fragmentco{x}{x'}$ and $\hy\in \fragmentco{y}{y'}$, then  $\A$ \emph{aligns} $X\position{\hx}$ to $Y\position{\hy}$,  denoted by $X\position{\hx} \ponto{\A} Y\position{\hy}$.
    If $X\position{\hx}\ne Y\position{\hy}$, then $\A$ \emph{substitutes}  $X\position{\hx}$ for $Y\position{\hy}$.
    If  $X\position{\hx}= Y\position{\hy}$, then $\A$ \emph{matches} $X\position{\hx}$ with $Y\position{\hy}$.
\end{itemize}
Insertions, deletions, and substitutions are jointly called character \emph{edits}.

The weighted edit distance of fragments $X[x\dd x')$ and $Y[y\dd y')$ with respect to a weight function $w: \wtype$ is defined as the minimum cost of an alignment:
\[\ed^w(X[x\dd x'),Y[y\dd y'))=\min_{\A \in \aa(X[x\dd x'),Y[y\dd y'))} \ed^w_{\A}(X[x\dd x'),Y[y\dd y')).\]
We often consider alignments of the entire string $X$ onto the entire string $Y$; we then used simplified notation including $\aa(X,Y)$, $\edwa{\A}{X}{Y}$, and $\ed^w(X,Y)$.

\begin{definition}[Forest alignment]\label{def:ta}
    Consider forests $F,G\in \FSigma$.
    An alignment $\A\in \aa(F\fragmentco{f}{f'},G\fragmentco{g}{g'})$ is a \emph{forest alignment} if it satisfies the following consistency condition: 
\begin{quote}
    For every two aligned characters $F[\hat{f}]\ponto{\A} G[\hat{g}]$, also $F[\match{F}{\hat{f}}]\ponto{\A} G[\match{G}{\hat{g}}]$.
\end{quote}
   \noindent
    We write $\ta(F\fragmentco{f}{f'},G\fragmentco{g}{g'})\sub \aa(F\fragmentco{f}{f'},G\fragmentco{g}{g'})$ for the set of forest alignments.
\end{definition}

\begin{remark}\label{rmk:spare}
Consider a forest alignment $\A\in \ta(F\fragmentco{f}{f'},G\fragmentco{g}{g'})$. Then,
\begin{itemize}
    \item $\A$ deletes every character $F[\hat{f}]$ with $\hat{f}\in \fragmentco{f}{f'}$ and  $\match{F}{\hat{f}}\notin \fragmentco{f}{f'}$,
    \item $\A$ inserts every character $G[\hat{g}]$ with $\hat{g}\in \fragmentco{g}{g'}$ and  $\match{G}{\hat{g}}\notin \fragmentco{g}{g'}$.
\end{itemize}
\end{remark}

\newcommand{\EPSigma}{\overline{\PSigma}}
\newcommand{\PW}{\tilde{w}}

Define $\EPSigma = \PSigma \cup \{\emptystring\}$ and a mapping
$\lambda : \EPSigma \to \eSigma$ such that $\lambda(\op_a)=\lambda(\cl_a)=a$ for each $a\in \Sigma$, and $\lambda(\emptystring)=\emptystring$.
For a weight function $w: \wtype$, we define a corresponding weight function $\PW: \smash{\EPSigma}^2 \to \mathbb{R}_{\ge 0}$ so that $\PW(p,q)=\frac12 w(\lambda(p),\lambda(q))$ for all $p,q\in \EPSigma$.
The \emph{cost} of a forest alignment $\A\in \ta(F\fragmentco{f}{f'},G\fragmentco{g}{g'})$ with respect to a weight function $w$
is defined as $\ted^w_{\A}(F\fragmentco{f}{f'},G\fragmentco{g}{g'}) := \ed^{\PW}_{\A}(F\fragmentco{f}{f'},G\fragmentco{g}{g'})$.
Moreover, we define \[\ted^w(F\fragmentco{f}{f'},G\fragmentco{g}{g'})=\min_{\A\in \ta(F\fragmentco{f}{f'},G\fragmentco{g}{g'})}\ted^w_\A(F\fragmentco{f}{f'},G\fragmentco{g}{g'}).\]
For a threshold $k\in \mathbb{R}_{\ge 0}$, we set 
    \[\ted^w_{\le k}(F\fragmentco{f}{f'},G\fragmentco{g}{g'}) = \begin{cases}
        \ted^w(F\fragmentco{f}{f'},G\fragmentco{g}{g'}) & \text{if }\ted^w(F\fragmentco{f}{f'},G\fragmentco{g}{g'})\le k,\\
        \infty & \text{otherwise.}
    \end{cases}
    \]
By \cref{rmk:spare}, the following value is non-negative for every $\A\in \ta(F\fragmentco{f}{f'},G\fragmentco{g}{g'})$:
\begin{align*}
&\ated^w_{\A}(F\fragmentco{f}{f'},G\fragmentco{g}{g'}) \coloneq \\ &\qquad\qquad \ted^w_{\A}(F\fragmentco{f}{f'},G\fragmentco{g}{g'}) - \! \sum_{\substack{\hat{f}\in [f\dd f'):\\ \match{F}{\hat{f}}\notin [f\dd f')}} \!\PW(F[\hat{f}],\emptystring) - \!
\sum_{\substack{\hat{g}\in [g\dd g'):\\ \match{G}{\hat{g}}\notin [g\dd g')}} \!\PW(\emptystring,G[\hat{g}]).
\end{align*}
We naturally generalize this value to $\ated^w\!(F\fragmentco{f}{f'},G\fragmentco{g}{g'})$.

\begin{observation}\label{obs:atedvsted}
   For all forests $F,G\in \FSigma$, fragments $F\fragmentco{f}{f'}$ and $G\fragmentco{g}{g'}$,
   and weight functions $w:\wtype$, we have $\ated^w(F\fragmentco{f}{f'},G\fragmentco{g}{g'})=\ted^w(F_{\fragmentco{f}{f'}},G_{\fragmentco{g}{g'}})$.
\end{observation}

%% file: src/klein_overview.tex
\newcommand{\dpt}{\mathsf{dp}}
\newcommand{\heavy}{\text{heavy}}
\newcommand{\sz}{\mathsf{size}}
\newcommand{\mgets}{\stackrel{\min}{\gets}}
\newcommand{\res}{\mathsf{result}}

In this section, we reinterpret Klein's algorithm~\cite{K98} and develop its $\Oh(nk^2\log n)$-time variant.

\subsection{Klein's Algorithm}

Klein's algorithm~\cite{K98} uses dynamic programming to compute $\ted(F_{[l_F\dd r_F)},G_{[l_G\dd r_G)})=\ated(F[l_F\dd r_F),G[l_G\dd r_G))$ for $\Oh(n\log n)$ selected fragments $F[l_F\dd r_F)$ of $F$ and all $\Oh(n^2)$ fragments $G[l_G\dd r_G)$ of $G$.
We modify the algorithm slightly so that the computed value $\ted(F[l_F\dd r_F),G[l_G\dd r_G))$ includes the costs of deleting the single parentheses of nodes straddling $F[l_F\dd r_F)$ and inserting the single parentheses of nodes straddling $G[l_G\dd r_G)$.

\SetKwFunction{Klein}{Klein}
\begin{algorithm}[H]
    \caption{$\protect\Klein(l_F,r_F,l_G,r_G)$: Klein's algorithm for computing the tree edit distance}\label{alg:Klein}
    \KwIn{Two fragments $F[l_F \dd r_F)$ and $G[l_G \dd r_G)$ of the input forests}
    \KwOut{Compute and store the value of $\dpt[l_F, r_F, l_G, r_G]$}
        \lIf{$l_F = r_F$}{\label{alg:klein:empty_f}%
            $\res \gets \sum_{i\in [l_G\dd r_G)}  \PW(\emptystring, G[i])$%
        }
        \lElseIf{$l_G = r_G$}{\label{alg:klein:empty_g}%
            $\res \gets \sum_{i\in [l_F\dd r_F)}  \PW(F[i], \emptystring)$%
        }
        \Else{
            $u_F = \node{F}{l_F}$, $v_F = \node{F}{r_F-1}$, $u_G = \node{G}{l_G}$, $v_G = \node{G}{r_G-1}$\;
            \If{$u_F \notin V_{F[l_F\dd r_F)}$ \KwSty{or} ${\boldmath (}\;v_F\in V_{F[l_F\dd r_F)}$ \KwSty{and} $\sz(u_F) \le \sz(v_F)\;{\boldmath)}$}{\label{alg:klein:upd_left}
                $\res \gets \dpt[l_F+1, r_F, l_G, r_G] + \PW(F[l_F], \emptystring)$\;\label{alg:klein:upd_left:del}
                $\res \mgets \dpt[l_F, r_F, l_G+1, r_G] + \PW(\emptystring, G[l_G])$\;\label{alg:klein:upd_left:ins}
                \If{$u_F \in V_{F[l_F\dd r_F)}$ \KwSty{and} $u_G \in V_{G[l_G\dd r_G)}$}{
                $\res \mgets \PW(F[l_F], G[l_G])+\dpt[l_F+1, c(u_F), l_G+1, c(u_G)] \allowbreak{} \qquad + \PW(F[c(u_F)], G[c(u_G)])  + \dpt[c(u_F)+1, r_F, c(u_G)+1, r_G]$\;\label{alg:klein:upd_left:match}%
                }
            }
            \Else{\label{alg:klein:upd_right}
                $\res \gets \dpt[l_F, r_F-1, l_G, r_G] +  \PW(F[r_F-1], \emptystring)$\;\label{alg:klein:upd_right:del}
                $\res \mgets \dpt[l_F, r_F, l_G, r_G-1] +  \PW(\emptystring, G[r_G-1])$\;\label{alg:klein:upd_right:ins}
                \If{$v_F\in V_{F[l_F\dd r_F)}$ \KwSty{and} $v_G\in V_{G[l_G\dd r_G)}$}{
                $\res \mgets \dpt[l_F, o(v_F), l_G, o(v_G)] + \PW(F[o(v_F)], G[o(v_G)]) \allowbreak{} \qquad + \dpt[o(v_F)+1, r_F-1, o(v_G)+1, r_G-1] + \PW(F[r_F-1], G[r_G-1])$\;\label{alg:klein:upd_right:match}%
                }
            }
        }
        $\dpt[l_F,r_F,l_G,r_G] \gets \res$
\end{algorithm}

The algorithm's implementation is provided in~\cref{alg:Klein}. We use an operator $x \mgets y$ that assigns $x \gets y$ if $y < x$.
We also remember which of these assignments were applied so that we can later trace back the optimal alignment based on this extra information stored.

In the corner case when $F[l_F\dd r_F)$ or $G[l_G\dd r_G)$ is empty, then the unique (and thus optimal) forest alignment pays for inserting or deleting all characters in the other fragment.
If both $F[l_F\dd r_F)$ and $G[l_G\dd r_G)$ are non-empty, the algorithm considers nodes $u_F=\node{F}{l_F}$, $u_G=\node{G}{l_G}$, $v_F=\node{F}{r_F-1}$, and $v_G=\node{G}{r_G-1}$.

Let us first suppose that $u_F$ and $v_F$ are contained in $F[l_F\dd r_F)$.
If the subtree of $v_F$ is at least as large as the subtree of $u_F$, that is, $\sz(v_F)\ge \sz(u_F)$, where $\sz(x)=c(x)-o(x)+1$, then algorithm considers three possibilities: $F[l_F]$ is deleted, $G[l_G]$ is inserted, or $F[l_F]$ is aligned with $G[l_G]$.
In the first two possibilities, we can pay for the deleted or inserted opening parenthesis and align the remaining fragments; see Lines~\ref{alg:klein:upd_left:del}--\ref{alg:klein:upd_left:ins}.
In the third possibility, the consistency condition in the definition of the forest alignments requires that $F[\match{F}{l_F}]$ is also aligned with $G[\match{G}{l_G}]$. 
In particular, the node $u_G$ needs to be contained in $G[l_G\dd r_G)$ so that $\match{G}{l_G}=c(u_G)$ and $\match{F}{l_F}=c(u_F)$.
Thus, we align $F(l_F\dd c(u_F))$ with $G(l_G\dd c(u_G))$, align $F(c(u_F)\dd r_F)$ with $G(c(u_G)\dd r_G)$,
and pay for aligning $u_F$ with $u_G$, that is $F[l_F]$ with $G[l_G]$ and $F[c(u_F)]$ with $G[c(u_G)]$; see Line~\ref{alg:klein:upd_left:match}.
If $\sz(v_F) < \sz(u_F)$, the algorithm handles $v_F$ and $v_G$ in a symmetric way; see Lines~\ref{alg:klein:upd_right}--\ref{alg:klein:upd_right:match}.

If $u_F\notin V_{F[l_F\dd r_F)}$, we follow Lines~\ref{alg:klein:upd_left}--\ref{alg:klein:upd_left:match} and effectively consider deleting $F[l_F]$ and inserting $G[l_G]$.
Else, if $v_F \notin V_{F[l_F\dd r_F)}$, we follow the symmetric Lines~\ref{alg:klein:upd_right}--\ref{alg:klein:upd_right:match}.

Zhang and Shasha~\cite{ZS89} use a very similar dynamic programming; in a baseline version, their algorithm always constructs the alignment from left to right instead of picking the side depending on $u_F$ and $v_F$.
Klein's optimization allows for an improved running time of $\Oh(n^3 \log n)$ instead of $\Oh(n^4)$.
In \cref{app:our_aj}, we recall the proof of the following lemma; the underlying insights will be necessary to analyze the optimized algorithm in \cref{sec:optimization}.

\begin{restatable}{lemma}{lemklein}\label{lem:klein}
    The recursive implementation of \cref{alg:Klein} visits $\Oh(n \log n)$ fragments~of~$F$.
\end{restatable}

\subparagraph*{Akmal and Jin's Algorithm}
Akmal and Jin~\cite{AJ21} reduce the time complexity of Klein's algorithm to $\Oh(nk^2 \log n)$ when computing $\ted_{\le k}(F,G)$.
Their solutions prunes some dynamic programming states so that, in the surviving states, the differences between the sizes of the subforests $F_{[l_F \dd r_F)}$ and $G_{[l_G \dd r_G)}$ is at most $k$.
This condition also holds for the difference between the sizes of $F_{[0 \dd r_F)}\setminus F_{[l_F \dd r_F)}$ and $G_{[0 \dd l_G)} \setminus G_{[l_G \dd r_G)}$, as well as between $F_{[l_F \dd |F|)} \setminus F_{[l_F \dd r_F)}$ and $G_{[l_G \dd |G|)} \setminus G_{[l_G \dd r_G)}$; see \cite[Lemma 12]{AJ21}.
As shown in \cite[Lemma 13]{AJ21}, for each subforest  $F_{[l_F \dd r_F)}$, there are $\Oh(k^2)$ subforests $G_{[l_G \dd r_G)}$ satisfying these three conditions. 
With a careful implementation, the total number of states becomes $\Oh(nk^2 \log n)$.

%% file: src/new_nk2_overview.tex
\newcommand{\width}{\mathsf{width}}
Our variant of Klein's algorithm simply prunes all states with $|l_F-l_G|>2k$ or $|r_F-r_G|> 2k$.
In other words, we modify \cref{alg:Klein} so that $\dpt[l_F,r_F,l_G,r_G]$ is set to $\infty$ if either condition holds, and the existing instructions are executed otherwise. 
This does not increase the number $\Oh(n\log n)$ of visited fragments $F[l_F\dd r_F)$; see \cref{lem:klein}.
For each of these fragments, trivially, at most $\Oh(k^2)$ fragments of $G$ survive pruning, so the total number of states is $\Oh(nk^2 \log n)$.
The correctness of the pruning rules follows form the fact that, for all forests $F,G$, the width of any alignment in $\ta(F,G)$ does not exceed twice its cost,
where the width of an alignment $\A\in \aa(X[x\dd x'),Y[y\dd y'))$ is defined as $\width(\A)=\max\{|\hx-\hy| : (\hx,\hy)\in \A\}$.

In \cref{app:our_aj}, we formalize this intuition using the notion of bounded forest alignments.

\begin{definition}
    Let us fix a threshold $k\in \Zp$ and consider fragments $F[f\dd f')$ and $G[g\dd g')$ of forests $F,G\in \FSigma$.
    We call a forest alignment $\A \in \ta(F[f\dd f'),G[g\dd g'))$ \emph{bounded} if $\width(\A)\le 2k$.
    The family of \emph{bounded forest alignments} is $\bta{k}(F[f\dd f'),G[g\dd g')) \subseteq \ta(F[f\dd f'),G[g\dd g'))$.
    For a weight function $w : \wtype$, we denote 
    \[\bted{k}^w(F[f\dd f'),G[g\dd g')) =\min_{\A \in \bta{k}(F[f\dd f'),G[g\dd g'))} \ted_{\A}^w(F[f\dd f'),G[g\dd g')).\]
\end{definition}

\begin{restatable}{lemma}{lemgood}\label{lem:good}
    The $\dpt$ values computed using the pruned version of \cref{alg:Klein} satisfy 
    \[\ted^w(F[l_F\dd r_F),G[l_G\dd r_G)) \le \dpt[l_F, r_F, l_G, r_G] \le \bted{k}^w(F[l_F\dd r_F),G[l_G\dd r_G)).\]
\end{restatable}

In particular, the $\dpt[0,|F|,0,|G|]$ entry stores a value between $\ted^w(F,G)$ and $\bted{k}^w(F,G)$. 
The following observation implies that this is enough to retrieve $\ted_{\le k}^w(F,G)$.
\begin{observation}\label{obs:good_is_good}
    If $\ted^w(F,G)\le k$, then $\ted^w(F,G)=\bted{k}^w(F,G)$.
\end{observation}
\begin{proof}
    Consider the optimal forest alignment $\A\in \ta(F,G)$ with $\ted^w_{\A}(F,G) \le k$.
    For each edge $(f,g)\to (f',g')$, either $f-g=f'-g'$ (if the edge is diagonal) or $\big|(f-g)-(f'-g')|=1$ and the edge cost is at least $\tfrac12$ (if the edge is vertical or horizontal).
    Since $(0,0)\in \A$ and the total cost of edges in $\A$ is at most $k$, every $(f,g)\in \A$ satisfies $|f-g|\le 2k$, i.e., $\width(\A)\le 2k$ and $\A \in \bta{k}(F,G)$.
    Hence, $\bted{k}^w(F,G)\le \ted^w_{\A}(F,G) = \ted^w(F,G)$.
    The converse inequality  $\ted^w(F,G)\le \bted{k}^w(F,G)$ holds trivially due to  $\bta{k}(F,G)\subseteq \ta(F,G)$.
\end{proof}

With minor implementation details needed to avoid logarithmic overheads for memoization using a sparse table $\dpt$ (Akmal and Jin~\cite{AJ21} ignore this issue), we achieve the following result.

\thmaj*

%% file: src/matrix_algorithm_overview.tex
In this section, we present an optimized version of our $\Oh(nk^2 \log n)$-time algorithm capable of exploiting certain repetitive structures within the input forests $F$ and $G$.
The following notion of \emph{free pairs} captures the structure that our algorithm is able to utilize.
\begin{definition}[Free pair, free block]\label{def:free}
    Consider forests $F,G\in \FSigma$ and a fixed threshold $k\in \Zp$.
    We call a pair of fragments $F[p_F \dd q_F)$ and $G[p_G \dd q_G)$ a \emph{free pair} if 
    \begin{itemize}
    \item $|p_F-p_G|\le 2k$, and
    \item there exists a balanced string $R\in \FSigma$ with $4k \le |R| < 8k$ such that $F[p_F-|R| \dd q_F+|R|)=R^{e+2}=G[p_G-|R| \dd q_G+|R|)$ holds for some integer exponent $e\in \Zp$.
    \end{itemize}
    For that free pair, we call the fragment $F[p_F \dd q_F)$ a \emph{free block} with period $R$ and exponent~$e$.
\end{definition}

\newcommand{\FB}{\mathbf{F}}
Our improved algorithm assumes that the input forests $F$ and $G$ are augmented with a collection $\FB$ of disjoint free blocks $F[p_F\dd q_F)$, each associated with the underlying period~$R$, exponent $e$, and the corresponding fragment $G[p_G\dd q_G)$.
The speed-up compared to the algorithm of \cref{sec:our_aj} is noticeable if the free blocks in $\FB$ jointly cover most of the characters of $F$, that is, the number of remaining \emph{non-free} characters is asymptotically smaller than $|F|$.

\begin{restatable}{theorem}{thmoptimization}\label{thm:optimization}
    There exists a deterministic algorithm that, given forests $F,G\in \FSigma$ of total length $n$, oracle access to a normalized weight function $w:\wtype$,
    a threshold $k\in \Zp$, and a collection $\FB$ of $t$ disjoint free blocks in $F$ such that $m$ characters of $F$ are not contained in any free block, computes $\ted^w_{\le k}(F,G)$ in $\Oh(n\log n + mk^2\log n + tk^3\log n)$ time.
\end{restatable}

Intuitively, the algorithm skips all the free blocks, which allows reducing the $\Oh(nk^2 \log n)$ running time of \cref{thm:aj} to $\Oh(mk^2 \log n)$, i.e., the algorithm needs to pay for non-free characters only.
Nevertheless, processing each free block takes $\Oh(k^3 \log n)$ extra time, and $\Oh(n\log n)$-time preprocessing time is still needed to avoid overheads for memoization.

\subparagraph*{Processing Free Blocks}
To understand our speed-up, let us consider a free pair with period $R$ and exponent $e$, that is, $F[p_F \dd q_F)=R^e=G[p_G \dd q_G)$.
We picked the parameters so that $|R|\ge 4k \ge \width(\A)+|p_F-p_G|$, and thus every alignment  $\A \in \bta{k}(F,G)$ aligns $F[p_F \dd q_F)=R^e$ with a fragment $G[p'_G\dd q'_G)$ contained in $G[p_G-|R|\dd q_G+|R|)=R^{e+2}$.
By the same argument, the image of every copy of $R$ within $F[p_F \dd q_F)=R^e$ is contained within the corresponding copy of $R^3$ within $G[p_G-|R|\dd q_G+|R|)=R^{e+2}$.
Moreover, when we align $F[p_F \dd q_F)$ to $G[p'_F\dd q'_F)\subseteq G[p_G-|R|\dd q_G+|R|)$, it suffices to partition $G[p'_F\dd q'_F)$ into $e$ fragments and \emph{independently} optimally align each copy of $R$ in $F[p_F \dd q_F)$ with the corresponding fragment of $G[p'_F\dd q'_F)$.
This is because $R$ is balanced, so every node of $R^e$ is contained within a single copy of $R$, and thus the consistency condition in \cref{def:ta} does not impose any constraints affecting multiple copies of $R$.

The optimal costs of aligning $R$ with relevant fragments of $R^3$ can be encoded in the following matrix, constructible in $\Oh(|R|^3 \log |R|)$ time using \cref{alg:Klein} (Klein's algorithm).

\begin{definition}\label{def:ted_matrix}
    For a balanced string $R\in \FSigma$, we define a matrix $M_R$ of size $(2|R|+1) \times (2|R|+1)$ with indices $i,j$ in the range $[-|R|\dd |R|]$ as follows:
    \[
        M_R[i,j] = \begin{cases}
            \ted^w(R, R^3[|R|+i\dd 2|R|+j)) & \text{if }|R|+i \le 2|R|+j,\\
            \infty & \text{otherwise.}
        \end{cases}
    \]
\end{definition}

In order to derive the optimal costs of aligning $R^e=F[p_F \dd q_F)$ with the relevant fragments of $R^{e+2}=G[p_G-|R|\dd q_G+|R|)$, we simply compute the $e$-th power of $M_R$ with respect to the min-plus product.
Formally, the min-plus product of matrices $A\in \mathbb{R}^{I\times J}$ and $B\in \mathbb{R}^{J\times K}$ is a matrix $C\in \mathbb{R}^{I\times K}$ such that $C[i,k]=\min_{j\in J} A[i,j]+B[j,k]$ for $(i,k)\in I\times K$.

For each free block $F[p_F \dd q_F)=R^e\in \FB$, we construct the matrix $M_R^e$ in $\Oh(k^3 \log n)$ time using \cref{alg:Klein} followed by fast exponentiation, which reduces to computing $\Oh(\log e)$ min-plus products.
We apply the matrix whenever we are tasked with filling a $\dpt[l_F,r_F,l_G,r_G]$ entry such that $F[p_F\dd q_F)$ is a prefix or a suffix of $F[l_F\dd r_F)$.
Formally, if $l_F\le p_F  < q_F = r_F$, then we use the following formula instead of following \cref{alg:Klein}:
\begin{equation}\label{eq:free_right}
    \dpt[l_F, r_F, l_G, r_G] \gets \min_{p_G' \in [p_F-|R|\dd p_F+|R|]} \dpt[l_F, p_F, l_G, p'_G] + M_R^e[p'_G-p_G,r_G-q_G].
\end{equation}
In the symmetric case of $l_F= p_F  < q_F \le r_F$, we apply
\begin{equation}\label{eq:free_left}
    \dpt[l_F, r_F, l_G, r_G] \gets \min_{q_G' \in [q_F-|R|\dd q_F+|R|]} \dpt[q_F, r_F, q'_G, r_G] + M_R^e[l_G-p_G,q'_G-q_G].
\end{equation}

In \cref{app:optimization}, we formalize the intuition above to prove that \cref{eq:free_right,eq:free_left} preserve the invariant of \cref{lem:good}, that is, \[\ted^w(F[l_F\dd r_F),G[l_G\dd r_G)) \le \dpt[l_F, r_F, l_G, r_G] \le \bted{k}^w(F[l_F\dd r_F),G[l_G\dd r_G)).\]
For \eqref{eq:free_right}, this boils down to the following lemma; the case of \eqref{eq:free_left} is symmetric.
\begin{restatable}{lemma}{lemtedmatrixupdate}\label{lem:ted_matrix_update}
    Consider fragments $F[l_F\dd r_F)$, $G[l_G\dd r_G)$ of forests $F,G\in \FSigma$, a normalized weight function $w : \eSigma^2 \to \mathbb{R}_{\ge 0}$, and a threshold $k\in \Zp$ such that $|r_F-r_G|\le 2k$.
    If there is a free pair $F[p_F\dd q_F)=G[p_G\dd q_G)=R^e$ such that $F[p_F\dd q_F)$ is a suffix of $F[l_F\dd r_F)$, then
    \begin{align*}&\bted{k}^w(F[l_F\dd r_F),G[l_G\dd r_G))\\
        &\qquad\ge \min_{p'_G\in [p_F-|R|\dd p_F+|R|]} \bted{k}^w(F[l_F\dd p_F),G[l_G\dd p'_G))+M_R^e[p'_G-p_G,r_G-q_G]\\
        &\qquad\ge \min_{p'_G\in [p_F-|R|\dd p_F+|R|]} \ted^w(F[l_F\dd p_F),G[l_G\dd p'_G))+M_R^e[p'_G-p_G,r_G-q_G]\\
        &\qquad\ge \ted^w(F[l_F\dd r_F),G[l_G\dd r_G)).
    \end{align*}
\end{restatable}

We further argue in \cref{app:optimization} that the recursive implementation of \cref{alg:Klein} augmented with the optimizations of \eqref{eq:free_right} and \eqref{eq:free_left} visits $\Oh(m \log n)$ fragments $F[l_F\dd r_F)$, where $m$ is the number of non-free characters, including $\Oh(t \log n)$ fragments $F[l_F\dd r_F)$ for which \eqref{eq:free_right} or \eqref{eq:free_left} apply.
Specifically, our optimized algorithm visits fragments $F[l_F\dd r_F)$ visited by the original \cref{alg:Klein} and satisfying the following additional property: every free block $F[p_F\dd q_F)\in \FB$ is either disjoint with $F[l_F\dd r_F)$ or contained in $F[l_F\dd r_F)$. 
Our implementation, formalized as \cref{alg:free_block}, takes $\Oh(n\log n)$ extra preprocessing time to list fragments visited by \cref{alg:Klein} and filter those satisfying the aforementioned property.

%% file: src/new_kernel_overview.tex
In this section, we outline our approach to strengthen \cref{thm:prev_kernel} into the following result:

\begin{restatable}{theorem}{thmkernel}\label{thm:our_kernel}
    There exists a linear-time algorithm that, given forests $F$, $G$ and an integer $k\in \Zp$,
    constructs forests $F'$, $G'$ of size $\Oh(k^5)$ such that $\ted_{\le k}^w(F',G')=\ted_{\le k}^w(F,G)$ holds for every normalized \emph{quasimetric} $w$ (weight function satisfying the triangle inequality), and a collection of $\Oh(k^3)$ disjoint free blocks in $F'$ with $\Oh(k^4)$ non-free characters. 
\end{restatable}

Compared to \cref{thm:prev_kernel}, we require the presence of $\Oh(k^3)$ free blocks that jointly capture all but $\Oh(k^4)$ characters of the output forest $F'$; consult \cref{def:free}.
At the very high level, the proofs of both \cref{thm:prev_kernel,thm:our_kernel} consist in three steps: decomposing the input forests $F$ and $G$ into several pieces, identifying pieces that can be matched exactly, and replacing (pairs of) identical pieces with smaller equivalent counterparts.

\subparagraph*{Forest Decompositions and Piece Matchings}

Following~\cite{DGHKS23}, we say that a \emph{piece} of a forest $F$ is a \emph{subforest}---a balanced fragment $F[i\dd j)$---or a \emph{context}---a pair of fragments $\langle F[i\dd i');F[j'\dd j)\rangle$ such that $F[i\dd j)$ is a tree and $F[i'\dd j')$ is balanced. 
We denote the set of pieces contained in a fragment $F[i\dd j)$ of $F$ by $\Pc(F[i\dd j))$; we set $\Pc(F)=\Pc(F[0\dd |F|))$.

In isolation from $F$, a context can be interpreted as a pair of non-empty strings $C=\langle C_L;C_R \rangle\in \PSigma^+ \times \PSigma^+$ such that $C_L\cdot C_R$ is a tree. 
The \emph{composition} of contexts $C,D$ results in a context $C\star D := \langle C_L\cdot D_L; D_R\cdot C_R \rangle$.
Moreover, the composition of a context $C$ and a forest $H$ results in a tree $C \star H := C_L\cdot H \cdot C_R$.
For any \emph{decomposition} of a forest $F$ into disjoint pieces, one can recover $F$ using the concatenation and composition operations.

We define the \emph{depth} of a context $C = \langle C_L;C_R \rangle$ to be the number nodes of $C_L\cdot C_R$ with the opening parenthesis in $C_L$ and the closing parenthsis in $C_R$.
Note that the depth of the context $C\star D$ is equal to the sum of the depths of $C$ and $D$.

The following notion formalizes the concept of a matching between pieces of $F$ and~$G$.

\begin{definition}\label{def:matching}
    For two forests $F$ and $G$ and a fixed threshold $k\in \Zp$, a \emph{piece matching} between $F$ and $G$ is a set of pairs $\M \subseteq \Pc(F)\times \Pc(G)$ such that:
    \begin{itemize}
        \item across all pairs $(f,g)\in \M$, the pieces $f\in \Pc(F)$ are pairwise disjoint, and
        \item there exists a forest alignment $\A \in \ta(F,G)$ of width at most $2k$ (i.e., $\A \in \bta{k}(F,G)$) that matches $f$ to $g$ perfectly for every $(f,g)\in \M$.
    \end{itemize}
\end{definition}

The kernelization algorithm behind \cref{thm:prev_kernel} repeatedly identifies a piece matching $\M$ of size $|\M|=\Oh(k)$ covering $\Omega(n)$ vertices of $F$ and replaces each pair of matching pieces $(f,g)\in \M$ with a pair of ``equivalent'' pieces $(f',g')$ of size $\Oh(k^4)$. 
After $\Oh(\log n)$ steps, this yields forests of size $\Oh(k^5)$. 
Our strategy relies on the following new result:

\begin{restatable}{theorem}{thmdecomposition}\label{thm:decomposition}
    There exists a linear-time algorithm that, given forests $F,G\in \FSigma$ and a threshold $k\in \Zz$, either certifies that $\ted(F,G)> k$ or constructs a size-$\Oh(k)$ piece matching $\M$ between $F$ and $G$ that leaves $\Oh(k^4)$ unmatched characters.
\end{restatable}

The proof of \cref{thm:decomposition}, presented in \cref{app:decomposition}, reuses a subroutine of \cite{DGHKS23} to construct a decomposition $\DC\subseteq \Pc(F)$ of $F$ into $\Oh(n/k^3)$ pieces of size $\Oh(k^3)$ each.
The next step is to build a piece matching $\M \subseteq \DC\times \Pc(G)$ that leaves at most $k$ pieces of $\DC$ unmatched.
We modify a dynamic-programming procedure from~\cite{DGHKS23} so that, additionally, the unmatched characters of $G$ form $\Oh(k)$ fragments.
As a result, even though the obtained matching is of size $|\M|=\Oh(n/k^3)$, it is possible to reduce its size to $\Oh(k)$.
For this, it suffices to repeatedly identify pairs of adjacent pieces $f,f'\in \Pc(F)$ matched to adjacent pieces $g,g'\in \Pc(G)$, and then replace these pieces with their unions $f\cup f'$ and $g\cup g'$ (as formalized in \cref{def:adjacent}, two disjoint pieces are adjacent if their union, consisting of the characters contained in at least one of these pieces, can be interpreted as a piece).

\subparagraph*{Periodic Blocks and Red Characters}
The main ingredient of our kernelization algorithm is a procedure that replaces a pair of matching pieces $(f,g)\in \M$ with a pair of ``equivalent'' smaller pieces $(f',g')$.
Unlike~\cite{DGHKS23}, where the goal was to reduce the piece size to $\Oh(k^4)$, we aim to identify $\Oh(k^2)$ disjoint free blocks with $\Oh(k^3)$ non-free nodes within the replacement piece $f'$.
Unfortunately, free blocks lack a canonical construction, and it would be tedious to maintain a specific selection while the forests change. 
Instead, we use the following notions:

\begin{definition}\label{def:bf}
    For a fixed threshold $k\in \Zp$, we say that a string $S\in \PSigma^*$ forms a \emph{periodic block} if it satisfies the following properties:
        \begin{itemize}
        \item $|S| \ge 42k$, that is, the fragment is of length at least $42k$, and
        \item $S$ has a string period of length at most $4k$ with equally many opening and closing parentheses.
    \end{itemize}
    For a string $T\in \PSigma^*$, we denote by $\Bk{T}$ the set of fragments of $T$ that are periodic blocks.
    For a context $C=\context{C_L}{C_R}$, we denote by $\Bk{C}$ the disjoint union of $\Bk{C_L}$ and $\Bk{C_R}$.
\end{definition}

We partition the characters of $T\in \PSigma^*$ into \emph{black} and \emph{red} based on the family $\Bk{T}$.
\begin{definition}\label{def:redblack}
    A character $T[i]$ of a string $T\in \PSigma^*$ is \emph{black} if there exists a periodic block $T[l\dd r)\in \Bk{T}$ such that $i\in [l+5k\dd r-5k)$. The remaining characters are \emph{red}.
    We denote by $\black{T}$ and $\red{T}$ the sets of black and red characters of $T$. 

    For a context $C=\context{C_L}{C_R}$, the sets $\black{C}=\black{C_L}\sqcup \black{C_R}$ and $\red{C}=\red{C_L}\sqcup \red{C_R}$ are defined as disjoint unions.
\end{definition}

\subparagraph*{Piece Reduction}
The following definitions in~\cite{DGHKS23} formalize the concept of equivalent pieces.

\begin{definition}[{\cite[Definition 3.4]{DGHKS23}}]
    For a threshold $k\in \Zz$ and a weight function $w$, forests $P,P'$ are called \emph{$\ted_{\le k}^w$-equivalent}
    if 
    \[\ted_{\le k}^w(F,G) = \ted_{\le k}^w(F[0 \dd l_F) \cdot P' \cdot F[r_F\dd |F|),G[0 \dd l_G) \cdot P' \cdot G[r_G\dd |G|))\]
    holds for all forests $F$ and $G$ with matching pieces $F[l_F\dd r_F)=P=G[l_G\dd r_G)$ satisfying $|l_F-r_G|\le 2k$.
\end{definition}

\begin{definition}[{\cite[Definition 3.9]{DGHKS23}}]
    For a threshold $k\in \Zz$ and a weight function $w$, contexts $P=\context{P_L}{P_R}$ and $P'=\context{P'_L}{P'_R}$ are called \emph{$\ted_{\le k}^w$-equivalent}
    if 
    \begin{multline*}\ted_{\le k}^w(F,G) = \ted_{\le k}^w(F[0 \dd l_F) \cdot P'_L \cdot F[l'_F \dd r'_F) \cdot P'_R \cdot F[r_F\dd |F|),\\ G[0 \dd l_G) \cdot P'_L \cdot G[l'_G\dd r'_G) \cdot P'_R \cdot G[r_G \dd |G|))\end{multline*}
    holds for all forests $F$ and $G$ with matching pieces $\context{F[l_F\dd l'_F)}{F[r'_F\dd r_F)}=P=\context{G[l_G\dd l'_G)}{\allowbreak G[r'_G\dd r_G)}$ satisfying $|l_F-l_G|\le 2k$ and  $|r_F-r_G|\le 2k$.
\end{definition}
In \cref{app:piece_reduction}, we slightly modify the arguments of~\cite{DGHKS23} to prove the following results:

\begin{restatable}[{see \cite[Lemma 3.17]{DGHKS23}}]{lemma}{lemforestreduction}\label{lem:alg:forest_reduction}
    There is a linear-time algorithm that, given a forest $P$ and a threshold $k\in \Zp$, computes a forest $P'$ with $|P'| \le |P|$ and $|\red{P'}| \le 158k^2$ such that $P$ and $P'$ are $\ted_{\le k}^w$-equivalent for every normalized quasimetric weight function $w$.
\end{restatable}

\begin{restatable}[{see \cite[Lemma 3.18]{DGHKS23}}]{lemma}{lemcontextreduction}\label{lem:alg:context_reduction}
    There is a linear-time algorithm that, given a context $P$ and a threshold $k\in \Zp$, computes a context $P'$ with $|P'| \le |P|$ and $|\red{P'}| \le 1152k^3$ such that $P$ and $P'$ are $\ted_{\le k}^w$-equivalent for every normalized quasimetric weight function $w$.
\end{restatable}

\subparagraph{Complete Kernelization Algorithm}
In \cref{app:kernel}, we combine the above ingredients to formally prove \cref{thm:our_kernel}.
Our procedure first applies \cref{thm:decomposition}.
Then, for every pair of matched pieces $(f,g)\in \M$, we use \cref{lem:alg:forest_reduction} or~\ref{lem:alg:context_reduction} (depending on piece type) to obtain an equivalent piece $P'$ with $\Oh(k^3)$ red characters.
The remaining (black) characters in $P'$ can be traced back to $\Oh(k^2)$ periodic blocks
and each periodic block within $P'$ yields a free block in the output forest $F'$ that covers the underlying black characters.

%% file: src/conclusions.tex
\thmmain*
\begin{proof}
First, suppose that the task is to compute $\ted^w_{\le k}(F,G)$ for a given threshold $k$.
In this case, we use the algorithm of \cref{thm:our_kernel}, resulting in a pair of forests $F',G'$ of size $\Oh(k^5)$ such that $\ted^{w}_{\le k}(F',G')=\ted^{w}_{\le k}(F,G)$, as well as a collection of $\Oh(k^3)$ free blocks in $F'$ with $\Oh(k^4)$ non-free characters.
Based on this, the algorithm of \cref{thm:optimization} computes $\ted^{w}_{\le k}(F',G')=\ted^{w}_{\le k}(F,G)$ in $\Oh(k^5\cdot \log k^5 + k^4\cdot k^2 \log k^5 + k^3\cdot k^3 \log k^5)=\Oh(k^6 \log k)$ time.
Including the $\Oh(n)$ running time of \cref{thm:our_kernel}, we get $\Oh(n+k^6 \log k)$ total time.

In the absence of a given threshold, we consider a geometric sequence of thresholds $(d_i)_{i\in \Zz}$, with $d_i = 2^i\cdot \lceil{(n/\log n)^{1/6}}\rceil$, and we compute $\ted^w_{\le d_i}(F,G)$ for subsequent $i\in \Zz$ until $\ted^w_{\le d_j}(F,G)\le d_j$ holds for some $j\in \Zz$, which indicates $\ted^w(F,G)=\ted^w_{\le d_j}(F,G)$.

Since $d_0 = \Oh((n/\log n)^{1/6})$, the initial iteration costs $\Oh(n+d_0^6 \log d_0) = \Oh(n)$ time. 
Consequently, if  $k\le d_0$, then the whole algorithm runs in $\Oh(n)$ time.

Due to $d_i \ge d_0 \ge (n/\log n)^{1/6}$, the running time of the $i$th iteration iteration is $\Oh(d_i^6 \log d_i)$.
This sequence grows geometrically, so the total running time of the algorithm is dominated by the running time of the last iteration, which is $\Oh(d_j^6 \log d_j)$.
If $k > d_0$, then $j>0$ and, since the algorithm has not terminated one iteration earlier, $d_j = 2d_{j-1} < 2k$.
Consequently, the overall running time is $\Oh(k^6 \log k)$ when $k>d_0$.
\end{proof}

%% file: src/klein.tex
In this section, we fill in the missing details from \cref{sec:our_aj}. 

\subsection{Klein's Algorithm}

Below, we explain why \cref{alg:Klein} can be implemented in $\Oh(n^3 \log n)$ time.
Recall that, for a node $u\in V_F$, we defined $\sz(u)=|F[o(u)\dd c(v)]|=c(v)-o(v)+1$.

\begin{definition}\label{def:heavy}
    We call a node $u$ of a forest $F$ \emph{heavy} if $\sz(u)\ge \sz(u')$ holds for every left sibling $u'$ of $u$ and $\sz(u)>\sz(u')$ holds for every right sibling $u'$ of $u$; otherwise, $u$ is \emph{light}.
    For a non-leaf node $v$, we denote the heavy child of $v$ by $\heavy_v$ (it is the child of $v$ with the largest subtree size; in case of a tie, the rightmost among the tied children).
    Moreover, we denote the heavy root of $F$ by $\heavy_\bot$, where $\bot$ can be interpreted as the virtual root of $F$.
\end{definition}

\lemklein*
\begin{proof}
    We create a mapping from the visited fragment to the vertices of $F$ as follows:
    For each visited fragment $F[l_F \dd r_F)$, we map it to the lowest node $v$ that encloses $F[l_F \dd r_F)$, that is $o(v) < l_F \le r_F \le c(v)$.
    If no such node exists, we set $v$ to be the virtual root $\bot$ of $F$.
    In other words, $v$ is the lowest common ancestor of $\node{F}{l_F-1}$ and $\node{F}{r_F}$, with the convention that it is equal to $\bot$ if $l_F=0$, $r_F=|F|$, or $\node{F}{l_F-1}$ and $\node{F}{r_F}$ have no common ancestor.

    For each non-leaf node $v$ (including the virtual root $\bot$), consider the following sequence of fragments starting from $F(o(v)\dd c(v))$ (equal to $F$ if $v=\bot$) and ending at $F(o(\heavy_v) \dd c(\heavy_v)]$:
    at each step, the fragment $F[l_F \dd r_F)$ shrinks by one character, either $F[l_F]$ or $F[r_F-1]$ depending on whether \cref{alg:Klein} uses the branch of Line~\ref{alg:klein:upd_left}~or~\ref{alg:klein:upd_right}.

    \begin{claim}
        The sequence constructed for $v$ is well-defined, that is, $F(o(\heavy_v) \dd c(\heavy_v)]$ can indeed be obtained by repeatedly applying the described successor rule.
    \end{claim}
    \begin{claimproof}
        We prove a stronger statement that $F[o(\heavy_v) \dd c(\heavy_v)]$ belongs to the sequence constructed for $v$. 
        The branch of Line~\ref{alg:klein:upd_left} applies to this fragment (with $u_F=v_F=\heavy_v\in V_{F[l_F\dd r_F)}$), and thus the next element in the sequence is $F(o(\heavy_v) \dd c(\heavy_v)]$.

        Consider a fragment $F[l_F \dd r_F)$ with $[o(\heavy_v) \dd c(\heavy_v)]\subsetneq [l_F \dd r_F) \subseteq (o(v)\dd c(v))$. It suffices to prove that its successor still contains $[o(\heavy_v) \dd c(\heavy_v)]$.
        This condition may only be violated if $l_F=o(\heavy_v)$ or $r_F-1=c(\heavy_v)$, so we consider these two cases.
        \begin{description}
            \item[$l_F=o(\heavy_v)$.] In this case, $u_F = \heavy_v \in V_{F[l_F\dd r_F)}$ whereas $v_F$ is a proper ancestor of $v$ located strictly to the right of $\heavy_v$. In particular, a right sibling of $\heavy_v$ is an ancestor of $v_F$, and thus $\sz(v_F) < \sz(\heavy_v) = \sz(u_F)$ holds by \cref{def:heavy}.
            Hence, the algorithm chooses the branch of Line~\ref{alg:klein:upd_right}, and thus the successor fragment $F[l_F\dd r_F-1)$ still contains $F[o(\heavy_v) \dd c(\heavy_v)]$.
            \item[$r_F-1=c(\heavy_v)$.] In this case, $v_F = \heavy_v \in V_{F[l_F\dd r_F)}$ whereas $u_F$ is a proper ancestor of $v$ located strictly to the left of $\heavy_v$. In particular, a left sibling of $\heavy_v$ is an ancestor of $v_F$, and thus $\sz(v_F) \le \sz(\heavy_v) = \sz(u_F)$ holds by \cref{def:heavy}.
            Hence, the algorithm chooses the branch of Line~\ref{alg:klein:upd_left}, and thus the successor fragment $F[l_F+1\dd r_F)$ still contains $F[o(\heavy_v) \dd c(\heavy_v)]$.\claimqedhere
        \end{description}
    \end{claimproof}

    We next show that the constructed sequences capture everything that \cref{alg:Klein} visits.

    \begin{claim}
    Every visited non-empty $F[l_F\dd r_F)$ fragment mapped to a node $v\in V_F\cup\{\bot\}$ belongs to the sequence constructed to $v$.
    \end{claim}
    \begin{claimproof}
    The claim is satisfied initially when the algorithm starts with $F=F(o(\bot)\dd c(\bot))$.

    For any execution of \cref{alg:Klein}, the recursive calls in Lines~\ref{alg:klein:upd_left:ins} and~\ref{alg:klein:upd_right:ins} keep $F[l_F\dd r_F)$ intact.
    The calls Lines~\ref{alg:klein:upd_left:del} and~\ref{alg:klein:upd_right:del} move to the next fragment in the sequence, except when $F[l_F \dd r_F)=F(o(\heavy_v) \dd c(\heavy_v)]$ is the last fragment in the sequence, in which case Line~\ref{alg:klein:upd_right:del} proceeds to $F(o(\heavy_v) \dd c(\heavy_v))$, the first fragment assigned to $\heavy_v$.

    The main challenge is to analyze the recursive calls in Lines~\ref{alg:klein:upd_left:match} and \ref{alg:klein:upd_right:match}.    
    Line~\ref{alg:klein:upd_left:match} makes two recursive calls.
    The first call involves $F(o(u_F)\dd c(u_F))$, which is the first fragment assigned to $u_F$.
    The second call involves $F(c(u_F)\dd r_F)$, which is either empty (if $u_F=v_F$) or, as we prove next, one of the subsequent elements in the sequence of $v$.
    In other words, we need to show that \cref{alg:Klein} would never use the branch of Line~\ref{alg:klein:upd_right} as $F[l_F\dd r_F)$ is gradually shrunk to $F(c(u_F)\dd r_F)$.
    For this, consider an intermediate fragment $F[l'_F\dd r_F)$ and the node $u'_F=\node{F}{l'_F}$.
    Note that $u'_F$ is an ancestor of $u_F$ and thus $\sz(u'_F)\le \sz(u_F) \le \sz(v_F)$.
    Since also $v_F\in V_{F[l'_F\dd r_F)}$, the algorithm indeed chooses the branch of Line~\ref{alg:klein:upd_left} independently of whether $u'_F\in V_{F[l'_F\dd r_F)}$.

    The analysis of the recursive calls in Line~\ref{alg:klein:upd_right:match} is mostly symmetric.
    The second call involves  $F(o(v_F)\dd c(v_F))$, which is the first fragment assigned to $v_F$.
    The first call involves $F[l_F\dd o(v_F))$, which is either empty (if $u_F=v_F$) or, as we prove next, one of the subsequent elements in the sequence of $v$.
    In other words, we need to show that \cref{alg:Klein} would never use the branch of Line~\ref{alg:klein:upd_left} as $F[l_F\dd r_F)$ is gradually shrunk to $F[l_F\dd o(v_F))$.
    For this, consider an intermediate fragment $F[l_F\dd r'_F)$ and the node $v'_F=\node{F}{r'_F-1}$.
    Note that $v'_F$ is an ancestor of $v_F$ and thus $\sz(v'_F)\le \sz(v_F) < \sz(u_F)$.
    Since also $u_F\in V_{F[l_F\dd r'_F)}$, the algorithm indeed chooses the branch of Line~\ref{alg:klein:upd_left} independently of whether $v'_F\in V_{F[l_F\dd r'_F)}$.
    \end{claimproof}

    Finally, observe that the number of fragments mapped to $v$ is $\sz(v) - \sz(\heavy_v)$,
    which is $2$ plus the total size of all the siblings of $\heavy_v$, i.e., $2$ plus the total size of all the light children on $v$.
    The sum of $\sz(u)$ over all light vertices is equal to twice the sum of the number of light ancestors for each vertex.
    Since every root-to-leaf path contains at most $\log n$ light nodes~\cite{10.1137/0213024}, the summation is bounded by $\Oh(n \log n)$.
\end{proof}

%% file: src/new_nk2.tex
For reference, we provide \cref{alg:new_nk2logn} as the modified version of \cref{alg:Klein}.

\SetKwFunction{NewDP}{BoundedKlein}
\begin{algorithm}
    \caption{$\protect\NewDP(k,l_F,r_F,l_G,r_G)$: Our new pruning rule for Klein's algorithm}
    \label{alg:new_nk2logn}
    \KwIn{A threshold $k\in \Zp$ and fragments $F[l_F \dd r_F)$ and $G[l_G \dd r_G)$ of the input forests}
    \KwOut{Compute and stores the value of $\dpt[l_F, r_F, l_G, r_G]$}
        \If{$|l_F - l_G| \le 2k$ \KwSty{and} $|r_F - r_G| \le 2k$}{
            $\Klein(l_F,r_F,l_G,r_G)$ \tcp*[r]{Call Klein's update rule}
        }
        \Else{
            $\dpt[l_F, r_F, l_G, r_G] \gets \infty$ \tcp*[r]{Set the value to $\infty$}\label{ln:infty}
        }
\end{algorithm}

\lemgood*
\begin{proof}
    Let us first prove that $\ted^w(F[l_F\dd r_F),G[l_G\dd r_G))\le \dpt[l_F, r_F, l_G, r_G]$.
    This already follows from the correctness of \cref{alg:Klein}, but we provide a full proof for completeness since our formulation of \cref{alg:Klein} is slightly different than in the literature.
    We consider several cases depending on where the $\dpt[l_F, r_F, l_G, r_G]$ values originates from.
    \begin{itemize}
        \item Line~\ref{ln:infty}. In this case, trivially $\dpt[l_F, r_F, l_G, r_G]=\infty \ge \ted^w(F[l_F\dd r_F),G[l_G\dd r_G))$.
        \item Line~\ref{alg:klein:empty_f}. In this case, $F[l_F\dd r_F)$ is empty and there exist a forest alignment that simply inserts every character in $G[l_G\dd r_G)$. Its cost, computed correctly in Lines~\ref{alg:klein:empty_f}, is an upper bound on $\ted^w(F[l_F\dd r_F),G[l_G\dd r_G))$.
        \item Line~\ref{alg:klein:empty_g}. This case is symmetric to the previous one.
        \item Line~\ref{alg:klein:upd_left:del}. In this case, pick an optimal forest alignment $\A\in \ta(F(l_F\dd r_F),G[l_G\dd r_G))$. 
        Prepending an edge $(l_F,l_G)\to (l_F+1,l_G)$, we can extend it to a forest alignment $\B \in \ta(F[l_F\dd r_F),G[l_G\dd r_G))$. Thus, by the inductive assumption
        \begin{align*}
        \dpt[l_F,r_F,l_G,r_G] &= \dpt[l_F+1,r_F,l_G,r_G]+\PW(F[l_F],\emptystring) \\
        &\ge \ted^w(F(l_F\dd r_F),G[l_G\dd r_G))+\PW(F[l_F],\emptystring) \\
        & = \ted^w_{\A}(F(l_F\dd r_F),G[l_G\dd r_G))+\PW(F[l_F],\emptystring) \\
        & = \ted^w_{\B}(F[l_F\dd r_F),G[l_G\dd r_G)) \\
        & \ge \ted^w(F[l_F\dd r_F),G[l_G\dd r_G)).
        \end{align*}
        \item Line~\ref{alg:klein:upd_left:del}. In this case, pick an optimal forest alignment $\A\in \ta(F[l_F\dd r_F),G(l_G\dd r_G))$. 
        Prepending an edge $(l_F,l_G)\to (l_F,l_G+1)$, we can extend it to a forest alignment $\B \in \ta(F[l_F\dd r_F),G[l_G\dd r_G))$. Thus, by the inductive assumption
        \begin{align*}
        \dpt[l_F,r_F,l_G,r_G] &= \dpt[l_F,r_F,l_G+1,r_G]+\PW(\emptystring,G[l_G]) \\
        &\ge \ted^w(F[l_F\dd r_F),G(l_G\dd r_G))+\PW(\emptystring,G[l_G]) \\
        & = \ted^w_{\A}(F[l_F\dd r_F),G(l_G\dd r_G))+\PW(\emptystring,G[l_G])  \\
        & = \ted^w_{\B}(F[l_F\dd r_F),G[l_G\dd r_G)) \\
        & \ge \ted^w(F[l_F\dd r_F),G[l_G\dd r_G)).
        \end{align*}
        \item Line~\ref{alg:klein:upd_left:match}. In this case, we consider optimal forest alignments $\A\in \ta(F(l_F\dd c(u_F)),\allowbreak G(l_G\dd c(u_G)))$ and $\A'\in \ta(F(c(u_F)\dd r_F),G(c(u_G)\dd r_G))$.
        We concatenate the edge $(l_F,l_G)\to (l_F+1,l_G+1)$, the alignment $\A$, the edge $(c(u_F),c(u_G))\to (c(u_F)+1,c(u_G)+1)$, and the alignment $\A'$ to obtain a forest alignment $\B \in \ta(F[l_F\dd r_F),G[l_G\dd r_G))$.
        Note that $\B$ matches the parentheses of $u_F$ to the corresponding parentheses of $u_G$, which satisfies the consistency condition of \cref{def:ta}.
        For the remaining diagonal edges in $\B$, the relevant subalignment $\A$ or $\A'$ already satisfies the consistency condition, and hence so does $\B$.
        Thus, by the inductive assumption
        \begin{align*}
            \dpt[l_F,l_G,r_F,r_G] &= \PW(F[l_F], G[l_G])+\dpt[l_F+1, c(u_F), l_G+1, c(u_G)] \\
            & \quad + \PW(F[c(u_F)], G[c(u_G)])  + \dpt[c(u_F)+1, r_F, c(u_G)+1, r_G] \\
            &\ge \PW(F[l_F], G[l_G])+\ted^w(F(l_F\dd c(u_F)),G(l_G\dd c(u_G))) \\
            & \quad + \PW(F[c(u_F)], G[c(u_G)])  + \ted_{\A'}^w(F(c(u_F)\dd r_F),G(c(u_G)\dd r_G))\\
            &= \PW(F[l_F], G[l_G])+\ted^w(F(l_F\dd c(u_F)),G(l_G\dd c(u_G))) \\
            & \quad + \PW(F[c(u_F)], G[c(u_G)])  + \ted_{\A'}^w(F(c(u_F)\dd r_F),G(c(u_G)\dd r_G))\\
            &= \ted^w_{\B}(F[l_F\dd r_F),G[l_G\dd r_G)) \\
            & \ge \ted^w(F[l_F\dd r_F),G[l_G\dd r_G)).
        \end{align*}
        \item Lines~\ref{alg:klein:upd_right:del},~\ref{alg:klein:upd_right:ins}, and~\ref{alg:klein:upd_right:match}. These three cases are symmetric to the previous three, and the arguments are analogous.
    \end{itemize}

    To prove the inequality $\dpt[l_F,r_F,l_G,r_G]\le \bted{k}^w(F[l_F\dd r_F),G[l_G\dd r_G))$, consider an optimal bounded forest alignment $\A \in \bta{k}(F[l_F\dd r_F),G[l_G\dd r_G))$ so that $\ted_{\A}^w(F[l_F\dd r_F),\allowbreak G[l_G\dd r_G))=\bted{k}^w(F[l_F\dd r_F),G[l_G\dd r_G))$.
    Since $\width(\A)\le 2k$, we have $|l_F - l_G| \le 2k$ and $|r_F - r_G| \le 2k$.
    If $F[l_F\dd r_F)$ or $G[l_G\dd r_G)$ is empty, then $\A$ must insert or delete every character in the other fragment. This is covered by Lines~\ref{alg:klein:empty_f} and~\ref{alg:klein:empty_g}.
    Otherwise, we consider the first or the last edge of $\A$ depending on which branch \cref{alg:Klein} follows.   
    If the algorithm follows the branch of Line~\ref{alg:klein:upd_left}, we consider three case:
    \begin{itemize}
        \item The first edge of $\A$ is $(l_F,l_G)\to (l_F+1,l_G)$, that is, $\A$ deletes $F[l_F]$.
        Then, the subpath of $\A$ from $(l_F+1,l_G)$ to $(r_F,r_G)$ is a bounded forest alignment. 
        In this case, due to Line~\ref{alg:klein:upd_left:del} and the inductive hypothesis, we have
        \begin{align*}
            \dpt[l_F,l_G,r_F,r_G] &\le \dpt[l_F+1,r_G,l_G,r_G]+\PW(F[l_F],\emptystring) \\
            & \le \bted{k}^w(F[l_F+1\dd r_F),G[l_G\dd r_G))+\PW(F[l_F],\emptystring) \\
            &\le \ted^w_{\A}(F[l_F+1\dd r_F),G[l_G\dd r_G))+\PW(F[l_F],\emptystring) \\
            &= \ted^w_{\A}(F[l_F\dd r_F),G[l_G\dd r_G)) \\
            &=  \bted{k}^w(F[l_F\dd r_F),G[l_G\dd r_G)).
        \end{align*}
        \item The first edge of $\A$ is $(l_F,l_G)\to (l_F,l_G+1)$, that is, $\A$ inserts $G[l_G]$.
        Then, the subpath of $\A$ from $(l_F,l_G+1)$ to $(r_F,r_G)$ is a bounded forest alignment. 
        In this case, due to Line~\ref{alg:klein:upd_left:ins} and the inductive hypothesis, we have
        \begin{align*}
            \dpt[l_F,l_G,r_F,r_G] &\le \dpt[l_F,r_G,l_G+1,r_G]+\PW(\emptystring,G[l_G]) \\
            & \le \bted{k}^w(F[l_F\dd r_F),G[l_G+1\dd r_G))+\PW(\emptystring,G[l_G]) \\
            &\le \ted^w_{\A}(F[l_F\dd r_F),G[l_G+1\dd r_G))+\PW(\emptystring,G[l_G]) \\
            &= \ted^w_{\A}(F[l_F\dd r_F),G[l_G\dd r_G)) \\
            &=  \bted{k}^w(F[l_F\dd r_F),G[l_G\dd r_G)).
        \end{align*}
        \item Otherwise, the first edge of $\A$ must be $(l_F,l_G)\to (l_F+1,l_G+1)$.
        In this case, by \cref{def:ta}, $\A$ must also align $F[\match{F}{l_F}]$ and $G[\match{F}{l_G}]$.
        In particular, we must have $u_F=\node{F}{l_F}\in V_{F[l_F\dd r_F)}$ and $v_F=\node{G}{l_G}\in V_{G[l_G\dd r_G)}$,
        and the alignment $\A$ must also contain an edge $(c(u_F),c(u_G))\to (c(u_F)+1,c(u_G)+1)$.
        Then, the subpaths of $\A$ from $(l_F+1,l_G+1)$ to $(c(u_F),c(u_G))$  and from $(c(u_F)+1,c(u_G)+1)$ to $(r_F,r_G)$ are bounded forest alignments.
        This is because no node of $F$ simultaneously exits $F(o(u_F)\dd c(u_F))$ and enters $F(c(u_F)\dd r_F)$;
        consequently, the consistency condition of \cref{def:ta} does not impose any constraint on the considered subalignments that $\A$ does not satisfy already. 
        In this case, due to Line~\ref{alg:klein:upd_left:match} and the inductive hypothesis,
        \begin{align*}
            \dpt[l_F,l_G,r_F,r_G] &\le \PW(F[l_F],G[l_G])+\dpt[l_F+1,c(u_F),l_G+1,c(u_G)]\\
            &\quad + \PW(F[c(u_F)],G[c(u_G)])+\dpt[c(u_F)+1,r_F,c(u_G)+1,r_G] \\
            & \le \PW(F[l_F],G[l_G])+\bted{k}^w(F(l_F\dd c(u_F)),G(l_G\dd c(u_G)))\\
            &\quad + \PW(F[c(u_F)],G[c(u_G)])+\bted{k}^w(F(c(u_F)\dd r_F),G(c(u_G)\dd r_G))\\
            & \le \PW(F[l_F],G[l_G])+\ted_{\A}^w(F(l_F\dd c(u_F)),G(l_G\dd c(u_G)))\\
            &\quad + \PW(F[c(u_F)],G[c(u_G)])+\ted_{\A}^w(F(c(u_F)\dd r_F),G(c(u_G)\dd r_G))\\
            &= \ted^w_{\A}(F[l_F\dd r_F),G[l_G\dd r_G)) \\
            &=  \bted{k}^w(F[l_F\dd r_F),G[l_G\dd r_G)).
        \end{align*}
    \end{itemize}
    If the algorithm follows the branch of Line~\ref{alg:klein:upd_right}, then a symmetric argument shows that  
    Lines~\ref{alg:klein:upd_right:del},~\ref{alg:klein:upd_right:ins}, and~\ref{alg:klein:upd_right:match}
    jointly guarantee that $\dpt[l_F,l_G,r_F,r_G] \le \ted^w_{\A}(F[l_F\dd r_F),G[l_G\dd r_G))=\bted{k}^w(F[l_F\dd r_F),G[l_G\dd r_G))$.
\end{proof}

\thmaj*
\begin{proof}
By \cref{lem:good}, applying \cref{alg:new_nk2logn} results in $\dpt[l_F,r_F,l_G,r_G]$ values satisfying the invariant $\ted^w(F[l_F\dd r_F),G[l_G\dd r_G)) \le \dpt[l_F, r_F, l_G, r_G] \le \bted{k}^w(F[l_F\dd r_F),G[l_G\dd r_G))$.
In particular, $\ted^w(F,G) \le \dpt[0,|F|,0,|G|] \le \bted{k}^w(F,G)$.
Thus, it suffices to return $\dpt[0,|F|,0,|G|]$ if this value does not exceed $k$ or $\infty$ otherwise.
If $\ted^w(F,G)> k$, then $\dpt[0,|F|,0,|G|] \ge \ted^w(F,G) \ge k$, so we correctly return $\infty$.
Otherwise, \cref{obs:good_is_good} guarantee that $\dpt[0,|F|,0,|G|]=\bted{k}^w(F,G)=\ted^w(F,G)\le k$ is also computed correctly.

By \cref{lem:klein}, the recursive implementation of \cref{alg:Klein} (and thus of \cref{alg:new_nk2logn}) visits $\Oh(n\log n)$ fragments of $F$.
Due to the pruning rule, at most $\Oh(k^2)$ fragments of $G$ are visited for each visited fragment of $F$ (the pruned states are not counted since they can be charged to the parent state).
Overall, this gives $\Oh(nk^2\log n)$ dynamic-programming states for which we need to store the computed $\dpt$ values.
After $\Oh(n)$ time preprocessing, which computes the prefix sums $\ed^{\PW}(F[0\dd f),\emptystring)$ and $\ed^{\PW}(\emptystring, G[0\dd g))$, each application of \cref{alg:new_nk2logn} can be implement in $\Oh(1)$ time excluding recursive calls.
Nevertheless, we need to use memoization to avoid processing the same dynamic-programming state multiple times, and this induces an $\Oh(\log n)$-factor overhead if we store the (sparse) $\dpt$ table using a standard dynamic dictionary.

To avoid this overhead, we utilize the structure discussed in the proof of \cref{lem:klein}.
For each non-leaf node $v$ (including the virtual root $\bot$), we store the sequence of fragments $F[l_F\dd r_F)$ assigned to the node $v$.
For each of these fragments, we store the $\dpt[l_F,r_F,l_G,r_G]$ values in a two-dimensional table indexed by $l_G\in [l_F-2k\dd l_F+2k]$ and $r_G\in [r_G-2k\dd r_G+2k]$.
Overall, this yields a space complexity of $\Oh(nk^2 \log n)$.
The index of any fragment $F[l_F\dd r_F)$ assigned to $v$ within the sequence of fragments assigned to $v$ can be computed in $\Oh(1)$ time: it is equal to $\sz(v)-2-(r_F-l_F)$, where $\sz(\bot)=|F|+2$.
Thus, as long as we keep track of the nodes assigned to each non-empty fragment $F[l_F\dd r_F)$ (the proof of \cref{lem:klein} already explains how to do so), we can access the corresponding two-dimensional table in $\Oh(1)$ time.
Empty fragments $F[l_F\dd r_F)$ can be handled using a separate table indexed by $l_F=r_F$.

This yields an $\Oh(nk^2 \log n)$-time implementation of the algorithm computing $\ted_{\le k}^w(F,G)$ for a given threshold $k$.
In the absence of a given threshold, we consider a geometric sequence of thresholds $(d_i)_{i\in \Zz}$ with $d_i=2^i$ and compute $\ted^w_{\le d_i}(F,G)$ for subsequent $i\in \Zz$ until $\ted^w_{\le d_j}(F,G)\le d_j$ holds for some $j\in \Zz$, which indicates that $\ted^w(F,G)=\ted^w_{\le d_j}(F,G)$.
This happens for $j=\lceil \log \ted^w(F,G)\rceil$. The running times of subsequent iterations form a geometric progression dominated by $\Oh(nd_j^2 \log n)=\Oh(nk^2\log n)$ as long as $F\ne G$.
\end{proof}

%% file: src/matrix_algorithm.tex
In this section, we fill in the missing details from \cref{sec:optimization}.
Recall that the correctness of our optimization hinges on the following result.

\lemtedmatrixupdate*
\begin{proof}
    For every $t\in [0\dd e]$, let us denote $x_t=p_F+|R|t$ and $y_t=p_G+|R|t$.
    We prove the following auxiliary claim:
    \begin{claim}\label{clm:ted_matrix_update_proof_single_induction_step}
        For every $t\in [0\dd e)$ and $z_{t+1}\in [y_{t+1}-|R|\dd y_{t+1}+|R|]$, we have
        \begin{align*}&\bted{k}^w(F[l_F\dd x_{t+1}),G[l_G\dd z_{t+1}))\\
            &\qquad\ge \min_{z_t\in [y_t-|R|\dd y_t+|R|]} \bted{k}^w(F[l_F\dd x_t),G[l_G\dd z_t))+M_R[z_t-y_t,z_{t+1}-y_{t+1}]\\
            &\qquad\ge \min_{z_t\in [y_t-|R|\dd y_t+|R|]} \ted^w(F[l_F\dd x_t),G[l_G\dd z_t))+M_R[z_t-y_t,z_{t+1}-y_{t+1}]\\
            &\qquad \ge \ted^w(F[l_F\dd x_{t+1}),G[l_G\dd z_{t+1})).
        \end{align*}
    \end{claim}
    \begin{claimproof}
        Consider an optimal alignment $\A \in \bta{k}(F[l_F\dd x_{t+1}),G[l_G\dd z_{t+1}))$.
        Due to $x_t \in [x_0\dd x_{t+1}]\subseteq [l_F\dd x_{t+1}]$, the alignment $\A$ contains a pair $(x_t, z_t)$ such that $|z_t-x_t| \le \width(A)$.
        Moreover, by the triangle inequality, $|z_t-y_t| \le |z_t-x_t|+|x_t-y_t| \le \width(A)+|p_F-p_G| \le 4k \le |R|$.
        Since $F[x_t\dd x_{t+1})=R$ is balanced, the subpaths of $\A$ from $(l_F,l_G)$ to $(x_t,z_t)$ and from $(x_t,z_t)$ to $(x_{t+1},z_{t+1})$ are bounded forest alignments.
        The latter one aligns $F[x_t\dd x_{t+1})=R$ with $G[z_{t}\dd z_{t+1})=R^3[|R|+z_t-y_t\dd 2|R|+z_{t+1}-y_{t-1})$, so its cost is at least $M_R[z_t-y_t,z_{t+1}-y_{t+1}]$.
        Consequently,
        \begin{align*}&\bted{k}^w(F[l_F\dd x_{t+1}),G[l_G\dd z_{t+1}))\\
            &\qquad = \ted^w_{\A}(F[l_F\dd x_{t+1}),G[l_G\dd z_{t+1})) \\
            &\qquad = \ted^w_{\A}(F[l_F\dd x_{t}),G[l_G\dd z_{t}))+\ted^w_{\A}(F[x_{t}\dd x_{t+1}),G[z_{t}\dd z_{t+1})) \\
            &\qquad \ge \bted{k}^w(F[l_F\dd x_t),G[l_G\dd z_t))+\ted^w(F[x_{t}\dd x_{t+1}),G[z_{t}\dd z_{t+1}))\\
            & \qquad = \bted{k}^w(F[l_F\dd x_t),G[l_G\dd z_t)) + M_R[z_t-y_t,z_{t+1}-y_{t+1}] \\
            & \qquad \ge  \min_{z_t\in [y_t-|R|\dd y_t+|R|]} \bted{k}^w(F[l_F\dd x_t),G[l_G\dd z_t))+M_R[z_t-y_t,z_{t+1}-y_{t+1}].
        \end{align*}
        This completes the proof of the first inequality.
        The second inequality holds trivially due to $\bta{k}(F[l_F\dd x_t),G[l_G\dd z_t))\subseteq \ta(F[l_F\dd x_t),G[l_G\dd z_t))$.
        As for the third inequality, pick any $z_t\in [y_t-|R|\dd y_t+|R|]$,
        an optimal alignment $\A \in \ta(F[l_F\dd x_t),G[l_G\dd z_t))$ of cost $\ted^w(F[l_F\dd x_t),G[l_G\dd z_t))$, and an optimal alignment $\A'\in \ta(F[x_t\dd x_{t+1}),G[z_{t}\dd z_{t+1}))$ of cost $\ted^w(F[x_{t}\dd x_{t+1}),G[z_{t}\dd z_{t+1}))$.
        Note that $F[x_t\dd x_{t+1})=R$ and $G[z_{t}\dd z_{t+1})=R^3[|R|+z_t-y_t\dd 2|R|+z_{t+1}-y_{t-1})$, so the cost of $\A'$ is equal to $M_R[z_t-y_t,z_{t+1}-y_{t+1}]$.
        The concatenation of alignments $\A$ and $\A'$ is a forest alignment in $\ta(F[l_F\dd x_{t+1}),G[l_G\dd z_{t+1}))$, and thus its cost is at least as large as $\ted^w(F[l_F\dd x_{t+1}),G[l_G\dd y_{t+1}))$.
    \end{claimproof}
    The lemma can be proved by repeatedly using \cref{clm:ted_matrix_update_proof_single_induction_step} for subsequent values of $t\in [0\dd e)$.
    In the final step, we also rely on the fact that $r_G\in [y_e-|R|\dd y_e+|R|]$.
    This holds because $y_e=q_G$, $|q_F-q_G|=|p_F-p_G|\le 2k$, $|r_F-r_G|\le 2k$, $x_{e}=r_F=q_F$, and $|R|\ge 4k$.
\end{proof}

\begin{lemma}\label{lem:klein_fragments}
    There exists an algorithm that, given a forest $F\in \FSigma$ of size $n$, in $\Oh(n \log n)$ time finds the set $\I$ of fragments of $F$ that the recursive formulation of \cref{alg:Klein} visits (for any non-empty $G\in \FSigma$). 
\end{lemma}
\begin{proof}
    We can virtually run Algorithm~\ref{alg:Klein} without giving $l_G$ and $r_G$ to the algorithm, and just keep track of the dp entries that the algorithm fetches.
    We can ignore lines~\ref{alg:klein:empty_f} and~\ref{alg:klein:empty_g} of the algorithm.
    The decision of which of the two branches in lines~\ref{alg:klein:upd_left} and~\ref{alg:klein:upd_right} to take is solely based on $l_F$ and $r_F$.
    In each branch, we have three update rules.
    The first rule removes a single character of $F$, which we make a recursive call for it in our algorithm.
    The second rule does not change the fragment, so we can ignore it.
    The third rule is applied based on some condition, which we assume is always satisfied and make recursive calls for it.
    \Cref{lem:klein} shows that our algorithm runs in $\Oh(n \log n)$ time and visits $\Oh(n \log n)$ fragments of~$F$.
\end{proof}

\SetKwFunction{FreeBlockR}{FreeBlock$_{\mathtt{Right}}$}
\SetKwFunction{FreeBlockL}{FreeBlock$_{\mathtt{Left}}$}

\thmoptimization*
\begin{proof}
We implement our final algorithm as \cref{alg:free_block}, with a subroutine $\FreeBlockR$, implemented in \cref{alg:free_block_update_right}, and a symmetric $\FreeBlockL$ responsible for \eqref{eq:free_right} and \eqref{eq:free_left}.

\begin{algorithm}
    \caption{$\protect\FreeBlockR(l_F,r_F,p_F,q_F)$: Free Block algorithm}\label{alg:free_block_update_right}
    \KwIn{A fragment $F[l_F \dd r_F)$ of $F$, and a free block $F[p_F \dd q_F)=G[p_G\dd q_G)=R^e$ with $l_F\le p_F < r_F=q_F$}
    \KwOut{Compute and store the value of $\dpt[l_F, r_F, l_G, r_G]$ for all valid $l_G, r_G$}
    \For{$l_G \in [l_F-2k \dd l_F+2k]$}{
        \For{$r_G \in [r_F-2k \dd r_F+2k]$}{
            $\dpt[l_F, r_F, l_G, r_G] \gets \min_{p_G' \in [p_F-|R|\dd p_F+|R|]} \dpt[l_F, p_F, l_G, p_G'] + M_R^e[r_G-q_G,p_G'-p_G]$.
        }
    }
\end{algorithm}

\begin{algorithm}
    \caption{Our optimized Bounded Tree Edit Distance algorithm}
    \label{alg:free_block}
    \KwIn{Forests $F$ and $G$, oracle access to a weight function $w:\wtype$, a threshold $k\in \Zp$,
    and a family $\FB$ of disjoint free blocks in $F$}
    \KwOut{The bounded tree edit distance $\ted^w_{\le k}(F,G)$}

    Run algorithm of Lemma~\ref{lem:klein_fragments} to find set $\I$ of fragments of $F$\;
    \ForEach{free block $F[p_F \dd q_F)=R^e \in \FB$}{
        Construct the matrix $M_R$ using \cref{alg:Klein}\;
        Construct the matrix $M_R^e$ using fast exponentiation and min-plus product\;
    }

    \For{$F[l_F\dd r_F) \in \I$ in the increasing order of $r_F-l_F$}{
        \If{$F[l_F \dd r_F)$ partially contains a free block in $\FB$}{
            \KwSty{continue}\;
        }
        \ElseIf{$F[l_F \dd r_F)$ ends with a free block $F[p_F\dd q_F)\in \FB$}{
            $\FreeBlockR(l_F,r_F,p_F,q_F)$\;
        }
        \ElseIf{$F[l_F \dd r_F)$ starts with a free block $F[p_F\dd q_F)\in \FB$}{
            $\FreeBlockL(l_F,r_F,p_F,q_F)$\;
        }
        \Else{
            \For{$l_G$ \KwSty{from} $\min(|G|, l_F+2k)$ \KwSty{downto} $\max(0, l_F-2k)$}{
                \For{$r_G$ \KwSty{from} $\max(l_G, r_F-2k)$ \KwSty{to} $\min(|G|, r_F+2k)$}{
                    $\Klein(l_F,r_F,l_G,r_G)$%
                }
            }
        }
    }
    \lIf{$\dpt[0, |F|, 0, |G|]\le k$}{\Return{$\dpt[0, |F|, 0, |G|]$}}
    \lElse{\Return{$\infty$}}
\end{algorithm}

    Algorithm~\ref{alg:free_block} computes selected $\dpt[l_F,r_F,l_G,r_G]$ values by simulating the bounded version of Klein's algorithm and taking advantage of the free blocks.
    We do not compute any values $\dpt[l_F,r_F,l_G,r_G]$ if $F[l_F\dd r_F)$ \emph{partially contains} some free block $F[p_F\dd q_F)\in \FB$, that is, $F[p_F\dd q_F)$ is neither contained in nor disjoint with $F[l_F\dd r_F)$.
    When $F[p_F\dd q_F)\in \FB$ is a prefix or a suffix of $F[l_F\dd r_F)$, we jump through the free block and compute the dp values using \eqref{eq:free_right} or \eqref{eq:free_left}.

    The algorithm first finds the set $\I$ of fragments of $F$ that Klein's algorithm visits, using Lemma~\ref{lem:klein_fragments}.
    Then it loops over the fragments in $\I$ in the increasing order of their length, since each update rule requires the dp values of shorter fragments.
    We skip each fragment $F[l_F\dd r_F)\in \I$ that partially contains some free block in $\FB$.
    If $F[l_F\dd r_F)\in \I$ starts or ends with a free block in $\FB$, we use the corresponding algorithm $\FreeBlockR$ or $\FreeBlockL$ which apply \eqref{eq:free_right} or \eqref{eq:free_left} to compute the dp values.
    Otherwise, we use the update rule of \cref{alg:Klein} to compute the dp values.
    
    \subparagraph*{Correctness}
    Our algorithm maintains the following invariant for each computed $\dpt$ value:
    \[\ted^w(F[l_F\dd r_F),G[l_G\dd r_G)) \le \dpt[l_F, r_F, l_G, r_G] \le \bted{k}^w(F[l_F\dd r_F),G[l_G\dd r_G)).\]
    \cref{lem:good} guarantees that the applications of $\Klein$ satisfy this invariant.
    For the applications of $\FreeBlockR$ and $\FreeBlockL$, this follows from \cref{lem:ted_matrix_update} and its symmetric counterpart.
    Nevertheless, the $\dpt[l_F,r_F,l_G,r_G]$ values that we skipped (because $F[l_F\dd r_F)$ partially contains  a free block in $\FB$) do not satisfy the invariant, so we need to argue that they are never used.
    Our update rules contain three types of recursive calls that refer to a different fragment of $F$:
    \begin{itemize}
        \item The first one is when we remove a single character from the left or right of the fragment $F[l_F\dd r_F)$ (Lines~\ref{alg:klein:upd_left:del} and \ref{alg:klein:upd_right:del}).
        Without loss of generality, we assume that we remove $F[l_F]$.
        If we reach a skipped fragment after removing the leftmost character, then the previous fragment was beginning with a free block in $\FB$, and we were not using Klein's update rule.
        \item The second type is when Klein's algorithm matches two vertices, and makes a recursive call to their subtrees and the rest of the fragments (Lines~\ref{alg:klein:upd_left:match} and \ref{alg:klein:upd_right:match} of \cref{alg:Klein}).
        Without loss of generality, consider Line~\ref{alg:klein:upd_left:match},
        where the two recursive calls are for $F(l_F\dd m_F)$ and $F(m_F\dd r_F)$ with $m_F = \match{F}{l_F}>l_F$.
        Since we reached the call to $\protect\Klein$ in \cref{alg:free_block}, the fragment $F[l_F\dd r_F)$ does not start with, end with, or partially contain any free block.
        Thus, every free block is either disjoint with $F[l_F\dd r_F)$ or contained in $F(l_F\dd r_F-1)$.
        If a free block is disjoint with $F[l_F\dd r_F)$, it is also disjoint with both $F(l_F\dd m_F)$ and $F(m_F\dd r_F-1)$.
        If a free block is contained in $F(l_F\dd r_F-1)$ then, since it is balanced, it cannot contain $F[m_F]$ (because $F[l_F]$ is not part of any free block), so it must be contained in $F(l_F\dd m_F)$ or $F(m_F\dd r_F-1)$.
        In either case, the free block cannot be partially contained in $F(l_F\dd m_F)$ or $F(m_F\dd r_F-1)$.
        \item The third type is the recursive call behind \eqref{eq:free_right} and \eqref{eq:free_left}, which skips the free block.
        Since the free blocks in $\FB$ are disjoint, the resulting fragment does not partially contain any free block in $\FB$.
    \end{itemize}

    \subparagraph*{Running Time}
    As for the running time, we first construct the set of fragments in $\Oh(n \log n)$ time.
    Then we sort them which can be a linear-time counting sort since the length of fragments is an integer in $[0 \dd n]$.
    For each free block in $\FB$, we compute $M_R^e$ which costs $\Oh(k^3\log n)$ per free block in $\FB$.
    Next, we need to show that the number of non-skipped fragments is $\Oh(m\log n)$, and the number of calls to $\FreeBlockR$ and $\FreeBlockL$ is $\Oh(t\log n)$.
    For each vertex $v$ in $F$, let $\sz'(v)$ be the number of non-free characters in its subtree and let $\sz''(v)$ be the number of free blocks in $\FB$ that are completely contained in the subtree of $v$.
    Similar to proof of \cref{lem:klein}, we map each visited fragment to the lowest enclosing vertex $v$.
    For each vertex $v$, the number of non-skipped fragments mapped to it is at most $\sz'(v)-\sz'(\heavy_v)$.
    If we compute the summation of this value over all vertices (and the virtual root of $F$), it results to sum of $\sz'(v)$ over all light vertices, including the virtual root of $F$.
    By charging each of these $\sz'(v)$ to the non-free characters in subtree of $v$, every non-free vertex is charged for each of its light ancestors, including the virtual root of $F$,
    and since each vertex has at most $\Oh(\log n)$ light ancestors, the summation is bounded by $\Oh(m\log n)$.

    The proof for number of calls to $\FreeBlockR$ and $\FreeBlockL$ is similar.
    We first observe that for each vertex $v$ and its heavy child $\heavy_v$, there is no free block in the subtree of $v$ that contains $o(\heavy_v)$ or $c(\heavy_v)$.
    This is because, for each free block in $\FB$, there is a copy of its period to its right that is not part of any free block.
    Now for each call to $\FreeBlockR$ and $\FreeBlockL$, we map the corresponding fragment to the lowest enclosing vertex $v$.
    Similar to previous case, each vertex has $\sz''(v)-\sz''(\heavy_v)$ calls mapped to it.
    By summing up over all vertices, we can show the sum is bounded by $\Oh(t\log n)$.

    The total running time is $\Oh(n\log n)$ for computing and sorting fragments in $\I$,
    $\Oh(t\cdot k^3\log n)$ for computing $M_R^e$ for all free pairs,
    $\Oh(k^2\cdot m\log n)$ for the number of fragments that we do not skip and run Klein's update rule for,
    and $\Oh(k^3\cdot t\log n)$ for the $\Oh(t\log n)$ calls to $\FreeBlockR$ and $\FreeBlockL$, each with $\Oh(k^3)$ cost.
    Summing up all these values gives us the final running time of $\Oh(n\log n + tk^3\log n + mk^2\log n)$.  
\end{proof}

%% file: src/new_kernel.tex
In this section, we fill in the missing details behind the proof of \cref{thm:our_kernel}.
Specifically, we prove \cref{thm:decomposition} in \cref{app:decomposition} and \cref{lem:alg:forest_reduction,lem:alg:context_reduction} in \cref{app:piece_reduction}.
Then, in \cref{app:kernel}, we derive \cref{thm:our_kernel} from these three ingredients.

\subsection{Forest Decompositions and Piece Matchings}\label{app:decomposition}
\input{src/forest_decomposition.tex}

\subsection{Piece Reductions}\label{app:piece_reduction}
In this section, we prove \cref{lem:alg:forest_reduction,lem:alg:context_reduction}.
Both statements rely on the following subadditivity property of red characters.
\begin{observation}\label{obs:red}
    For every two strings $X,Y\in \PSigma^*$, we have $|\red{X \cdot Y}| \le |\red{X}|+|\red{Y}|$.
    Consequently, the following properties hold for the sets of black and red characters:
    \begin{itemize}
        \item For every two forests $X$ and $Y$, $|\red{X \cdot Y}| \le |\red{X}|+|\red{Y}|$.
        \item For every context $X$ and forest $Y$, $|\red{X \star Y}| \le |\red{X}|+|\red{Y}|$.
        \item For every two contexts $X$ and $Y$, $|\red{X \star Y}| \le |\red{X}|+|\red{Y}|$.
    \end{itemize}
\end{observation}
\begin{proof}
It suffices to observe that in all three cases, every periodic block in $X$ or $Y$ remains a periodic block in the resulting string, forest, or context.
\end{proof}

Moreover, we show that the classification into red and black characters can be implemented efficiently,
with black characters originating from few periodic blocks.

\begin{lemma}\label{lem:findgood}
    There is a linear-time algorithm that, given a string $T\in \PSigma^*$ and the threshold~$k$, 
    computes the sets $\black{T}$ and $\red{T}$.
    Moreover, the algorithm constructs a set $\I \subseteq \Bk{T}$ of size $|\I|\le \frac{1}{2k}\red{T}$
    such that 
    \begin{equation}\label{eq:black}\black{T} = \bigcup_{T[l\dd r)\in \I} [l+5k\dd r-5k).\end{equation}
\end{lemma}
\begin{proof}
    We describe the algorithm here:
    \begin{enumerate}
        \item\label{step:1} Compute all maximal periodic fragments (runs) of $T$ using the linear-time algorithm~\cite{BIINTT17}.
        \item\label{step:2} Filter out all periodic fragments $T[l\dd r)$ whose length is less than $42k$ or whose shortest string period is longer than $4k$ or does not have equally many opening and closing parentheses.
        \item\label{step:5} The remaining fragments form $\I$, and $\black{T}$ is computed according to \eqref{eq:black}. 
    \end{enumerate}

    After step~\ref{step:2}, all remaining fragments are periodic blocks.
    Hence, $[l+5k\dd r-5k)\subseteq \black{T}$ thus holds for every $T[l\dd r)\in \I$.
    To prove the converse containment, consider a position $i\in \black{T}$.
    By definition, there is a periodic block $T[l\dd r)\in \Bk{T}$ such that $i\in [l+5k\dd r-5k)$.
    Let $T[l'\dd r')$ be its maximal extension preserving the shortest period.
    Since $T[l\dd r)$ has length at least $42k$ and a string period of length at most $4k$ with equally many opening and closing parentheses, so does $T[l'\dd r')$.
    Moreover, by maximality, $T[l'\dd r')$ is one of the runs of $T$.
    This run is computed in Step~\ref{step:1} and preserved in Step~\ref{step:2}.
    Due to $T[l'\dd r')\in \I$ and $[l+5k\dd r-5k) \subseteq [l'+5k\dd r'-5k)$, we conclude that \eqref{eq:black} indeed holds.

    It remains to prove that $|\I|\le \frac{1}{2k}\red{T}$.
    By \cite[Fact 2.2.4]{KociumakaPhD}, every two maximal periodic fragments $T[l_1\dd r_1)$ and $T[l_2\dd r_2)$ with periods $p_1$ and $p_2$, respectively, intersect in less than $p_1+p_2-\gcd(p_1,p_2)$ characters.
    Hence, for two distinct $T[l_1\dd r_1), T[l_2\dd r_2)\in \I$, the intervals $[l_1+4k\dd r_1-4k)$ and $[l_2+4k\dd r_2-4k)$ are disjoint.
    We can thus assign $2k$ red characters, $[l+4k\dd l+5k)\cup [r-5k\dd r-4k)$ to each $T[l\dd r)\in \I$ so that each red character is accounted for at most once.
\end{proof}

\begin{corollary}\label{lem:findgood_context}
    There is a linear-time algorithm that, given a context $P=\context{P_L}{P_R}$ and the threshold~$k$, 
    computes the sets $\black{P}$ and $\red{P}$.
    Moreover, the algorithm constructs a set $\I \subseteq \BP$ of size $|\I|\le \frac{1}{2k}\red{P}$
    such that \eqref{eq:black} holds.
\end{corollary}

\subsubsection{Forest Reduction}\label{app:forest_reduction}

From the definition of periodic blocks, we have an updated form of \cite[Lemma 3.7]{DGHKS23}.

\begin{lemma}[see {\cite[Lemma 3.7]{DGHKS23}}]\label{lem:horizontal_aperiodic_reduction}
    Let $k\in \Zp$ and let $P$, $P'$ be forests with at least $158k^2$ red characters each.
    Then, $P$ and $P'$ are $\ted_{\le k}^w$-equivalent for every normalized quasimetric~$w$.
    \end{lemma}
    \begin{proof}
        Suppose that $P$ occurs in forests $F$ and $G$ at positions $p_F$ and $p_G$, respectively, satisfying $|p_F-p_G|\le 2k$.
        Denote $F' = F[0\dd p_F)\cdot P' \cdot F[p_F+|P|\dd |F|)$  and $G' = G[0\dd p_G)\cdot P' \cdot G[p_G + |P|\dd |G|)$.

        Let $\A=(f_t,g_t)_{t=0}^m$ be an alignment such that $\ted^w(F,G)= \ted^w_\A(F,G) \le k$. 
        Moreover, let $(f_a,g_a)\in \A$ be the leftmost element of $\A$ such that $f_a \ge p_F$ or $g_a \ge p_G$,
        and let $(f_b,g_b) \in \A$ be the leftmost element of $\A$ such that $f_b \ge p_F+|P|$ and $g_b \ge p_G+|P|$.

        The proof of \cite[Lemma 3.7]{DGHKS23} derives $\ted^w(F',G')\le \ted^w(F,G)\le k$.
        The only part of the original argumentation that relies on a different set of assumptions about $P$ and $P'$ is the proof of the following claim.
        \begin{claim}\label{clm:horizontal_reduction}
            There exists $t\in [a\dd b]$ such that $f_t-g_t = p_F - p_G$.
            \end{claim}
            \begin{claimproof}
                Let us partition $P=F[p_F\dd p_F+|P|)$ into individual characters representing deletions or substitutions of $\A$
                and maximal fragments that $\A$ matches perfectly (to fragments of~$G$).
                By \cite[Fact 2.7]{DGHKS23}, the number of such fragments is at most $2k+1$, and they contain at least $158k^2-2k$ red characters in total.
                Hence, one of these fragments, denoted $R=F[r_F\dd r_F+|R|)$, has at least $\frac{158k^2-2k}{2k+1}\ge 52k$ red characters.
                Suppose that the fragment of $G$ matched perfectly to $R$ is $G[r_G\dd r_G+|R|)$.
                If $r_F-r_G = p_F-p_G$, the claim holds for $t\in [a\dd b]$ such that $(f_t,g_t)=(r_F,r_G)$. 
                Otherwise, we note that $R$ has period $q:=|(r_F-p_F)-(r_G-p_G)|\in [1\dd 4k]$.
                Hence, $R$ is a periodic fragment of $P$ with at least $52k$ red characters.
                By \cref{def:redblack}, a periodic block contains at most $10k$ red characters, so $R\notin \BF$.
                As $\per(R)\le q \le 4k$ and $|R|\ge 52k > 42k$, we conclude that the period $Q:= R[0\dd q)$ does not have equally many opening and closing parentheses.
   
                By symmetry (up to reversal), we assume without loss of generality that $Q$ has more opening than closing parentheses. 
                Consequently, the number of nodes exiting $Q$ is strictly larger than the number of nodes entering $Q$ and, by a simple inductive argument, than the number of nodes entering any finite prefix of $Q^\infty$.
                In particular, the number of nodes of $F$ exiting $F[r_F\dd r_F+q)$ is strictly larger than the number of nodes entering $F[r_F+q\dd r_F+|R|)$, and thus there is a node $u\in V_F$ such that $o(u)\in [r_F\dd r_F+q)$ yet $c(u)\ge r_F+|R|$. 
                Observe that the node $u$ satisfies $c(u)-o(u) \ge |R|-q \ge 52k-4k > 4k$.

                Let $v$ be the node of $G$ such that $o(u)-r_F = o(v)-r_G$. 
                Since $\A$ matches $F[r_F\dd r_F+|R|)$ perfectly with $F[r_G\dd r_G+|R|)$, it must match $F[o(u)]$ with $G[o(v)]$ and, by the consistency of forest alignments (\cref{def:ta}), $\A$ also matches $F[c(u)]$ with $G[c(v)]$.
                Consequently, both $|o(u)-o(v)|$ and $|c(u)-c(v)|$ are bounded by $\width(\A) \le 2\cdot \ted^w_{\A}(F,G) \le 2k$.

                Similarly, let $v'$ be the node of $G$ such that $o(u)-p_F = o(v')-p_G$.
                Since $F[p_F\dd p_F+|P|)=P=G[p_G\dd p_G+|P|)$ is balanced, we must also have $c(u)-p_F=c(v')-p_G$.
                Due to $o(u)-o(v')=c(u)-c(v')=p_F-p_G$, we conclude that $c(v')-o(v')=c(u)-o(u) > 4k$ and that both $|o(v)-o(v')|$ and $|c(v)-c(v')|$ are bounded by $2k+|p_F-p_G|\le 4k$.

                Now, observe that $c(v) \ge c(v')-4k > o(v')$ and $c(v') > o(v')+4k \ge o(v)$, which means that $[o(v)\dd c(v)]\cap [o(v')\dd c(v')]\ne \emptyset$ and thus $v$ is an ancestor of $v'$ or vice versa.
                In either case, we have $0 \ge (o(v')-o(v))\cdot (c(v')-c(v)) = ((o(u)-o(v))-(p_F-p_G))\cdot  ((c(u)-c(v))-(p_F-p_G))$.
                The value $(f_t-g_t)-(p_F-p_G)$ can change by at most one for subsequent indices $t$.
                The sign of this integer value is different when $(f_t,g_t)=(o(u),o(v))$ and $(f_t,g_t)=(c(u),c(v))$,
                so it must be equal to $0$ at some intermediate index $t\in [a\dd b]$.
        \end{claimproof}

        As discussed above, \cref{clm:horizontal_reduction} combined with the original arguments in the proof of  \cite[Lemma 3.7]{DGHKS23}, let us derive $\ted^w(F',G')\le \ted^w(F,G)$.
        The converse inequality follows by symmetry between $(F,G,P)$ and $(F',G',P')$.
\end{proof}

\lemforestreduction*
\SetKwFunction{ForestReduction}{ForestReduction}
\begin{algorithm}
    \caption{$\protect\ForestReduction(P,k)$}\label{alg:forest_reduction}
    \KwIn{a forest $P$}
    \KwOut{a forest $P'$ with $|\red{P'}| \le 158k^2$}
        \lIf{$|\red{P}| \le 158k^2$}{\Return{$P$}}
        \lElse{\Return $\op_a^{79k^2} \cl_a^{79k^2}$}
\end{algorithm}
\begin{proof}
    We apply the algorithm of \cref{lem:findgood} construct $\red{P}$.
    If $|\red{\I}| \le 158k^2$, then we can just return $P$.
    Otherwise, we construct a forest $P'$ with no black character and exactly $158k^2$ red characters.
    Applying \cref{lem:horizontal_aperiodic_reduction}, we get that $P$ and $P'$ are $\ted_{\le k}^w$-equivalent for every normalized quasimetric $w$.
    Since the algorithm of \cref{lem:findgood} runs in linear time, the whole algorithm runs in linear time.
\end{proof}

\subsubsection{Context Reduction}\label{app:context_reduction}

\begin{lemma}[{\cite[Lemma 3.10]{DGHKS23}}]\label{lem:vertical_periodic_reduction}
    Let $k\in \Zp$, let $Q$ be a context, and let $e,e'\in \mathbb{Z}_{\ge 6k}$.
    Then, $Q^e$ and $Q^{e'}$ are $\ted_{\le k}^w$-equivalent for every normalized weight function $w$.
\end{lemma}

We say that a context $P=\langle P_L; P_R\rangle$ avoids vertical $k$-periodicity if 
it cannot be expressed as $P = C \star Q^{6k+1}\star D$ for some contexts $C,Q,D$
satisfying $|Q| \in [1\dd 8k]$.

From the definition of periodic blocks, we have an updated form of \cite[Lemma 3.12]{DGHKS23}.

\begin{lemma}[see {\cite[Lemma 3.12]{DGHKS23}}]\label{lem:context_reduction}
    Let $k\in \Zp$, let $P=\langle P_L;P_R\rangle,P'=\langle P'_L;P'_R\rangle$ be contexts with at least $1152k^3$ red characters that avoid vertical $k$-periodicity and whose halves do not contain any balanced substring with more than $158k^2$ red characters.
    Then, $P$ and $P'$ are $\ted_{\le k}^w$-equivalent for every normalized weight function $w$.
\end{lemma}
\begin{proof}
    Suppose that $P$ occurs in forests $F$ and $G$ at nodes $u$ and $v$, respectively, satisfying $|o(u)-o(v)|\le 2k$ and  $|c(u)-c(v)|\le 2k$.
    Denote
    \begin{align*}
        F' = F[0 \dd o(u)) \cdot P'_L \cdot F[o(u)+|P_L| \dd c(u)-|P_R|] \cdot P'_R \cdot F(c(u) \dd |F|), \\ 
        G' = G[0 \dd o(v)) \cdot P'_R \cdot G[o(v)+|P_L| \dd c(v)-|P_R|] \cdot P'_R \cdot G(c(v) \dd |G|).
    \end{align*}
    Let $\A=(f_t,g_t)_{t=0}^m$ be an optimal forest alignment such that $\ted(F, G) = \ted_{\A} (F, G) \le k$.
    Moreover, let $(f_a,g_a)\in \A$ be the leftmost element of $\A$ such that $f_a \ge o(u)$ or $g_a \ge o(v)$,
    $(f_b,g_b)\in \A$ be the leftmost element of $\A$ such that $f_b \ge o(u)+|P_L|$ and $g_b \ge o(v)+|P_L|$,
    $(f_c,g_c)\in \A$ be the leftmost element of $\A$ such that $f_c > c(u)-|P_R|$ or $g_c > c(v)-|P_R|$,
    and let $(f_d, g_d)$ be the leftmost element of $\A$ such that $f_d > c(u)$ and $g_d > c(v)$.

    The proof of \cite[Lemma 3.12]{DGHKS23} derives $\ted^w(F',G')\le \ted^w(F,G)\le k$.
    The only part of the original argumentation that relies on a different set of assumptions about $P$ and $P'$ is the proof of the following claim.
    \begin{claim}\label{clm:vertical_aperiodic_reduction}
        There exist $t_L\in [a\dd b]$ such that $f_{t_L}-g_{t_L} = o(u)-o(u)$
        and $t_R\in [c\dd d]$ such that $f_{t_L}-g_{t_L} = c(u)-c(u)$.
    \end{claim}
    \begin{claimproof}
        By symmetry (up to reversal), we can focus without loss of generality on the first claim.
        Moreover, by symmetry between $F$ and $G$, we can assume without loss of generality that $f_a = o(u)$;
        in particular, this implies $f_a - g_a \ge o(u)-o(v)$.
        If there exists $t\in [a\dd b]$ such that $f_t - g_t \le o(u)-o(v)$,
        then, since $f_t-g_t$ may change by at most one for subsequent positions,
        there is also $t_L\in [a\dd b]$ such that $f_{t_L} - g_{t_L} = o(u)-o(v)$,
        Consequently, it remains to consider the case when $f_t - g_t > o(u)-o(v)$ holds for all $t\in [a\dd b]$.

        Let us express $P$ as a vertical composition of $e$ contexts $P=P_0\star \cdots \star P_{e-1}$, where $e$ is the depth of $P$.
        Observe that the occurrences of $P$ at node $u$ in $F$ and $v$ in $G$,
        for each $i\in [0\dd e)$, induce occurrences of $P_i$ at some nodes $u_i$ in $F$ and $v_i$ in $G$.
        Since $F(o(u_i)\dd o(u_i)+|P_{i,L}|)$ and $F(c(u_i)-|P_{i,R}|\dd c(u_i))$ are balanced,
        we conclude that $|\red{P_i}| \le 2\cdot (158k^2+1)\le 318k^2$.
        We can decompose $[0\dd e)$ into at most $k$ individual indices $i$ such that $\A$ does not match perfectly
        the occurrence of $P_i$ at $v_i$ and at most $k+1$ intervals $[i\dd i')$ such that $\A$ matches the occurrence $P_i\star \cdots \star P_{i'-1}$ at $v_i$ perfectly to a context in $G$.
        Let us choose such an interval $[i\dd i')$ maximizing $|\red{P_i\star \cdots \star P_{i'-1}}|$;
        this value is at least $\frac{1152k^3-k\cdot 318k^2}{k+1}\ge 417k^2$ (by \cref{obs:red}).
        Let $i''\in [i\dd i')$ be the maximum index such that $|\red{P_{i''}\star \cdots \star P_{i'-1}}|> 8k$;
        note that $|\red{P_i\star \cdots \star P_{i''-1}}|\ge 417k^2 - (318k^2+8k) \ge 91k^2$.

        For each $j\in [i\dd i'']$, denote by $u'_j$ the node matched with $v_j$ by $\A$.
        Note that $|o(u'_j)-o(u_j)|\le 4k$ and $|c(u'_j)-c(u_j)|\le 4k$.
        Moreover, $o(u'_j) - o(v_j) > o(u)-o(v)=o(u_j)-o(v_j)$ implies $o(u'_j) > o(u_j)$.
        Since $|P_j\star \cdots \star P_{i'-1}| > 8k$, we conclude that $u'_j = u_{j'}$ for some $j' \in [j\dd i')$.
        Moreover, if $j>i$, then $u'_j$ must be a child of $u'_{j-1}$.
        Hence, there exists $\delta > 0$ such that $u'_{j}=u_{j+\delta}$ holds for all $j\in [i\dd i'']$.
        For $j\in [i\dd i'')$, this implies $P_j = P_{j+\delta}$ and that both halves of $P_j\star \cdots \star P_{j+\delta-1}$
        are of length at most $4k$. 
        In particular, if we define $Q=P_i\star \cdots \star P_{i+\delta-1}$, then,
        due to $\frac{|P_i\star \cdots \star P_{i''-1}|}{8k} \ge \frac{91k^2}{8k} \ge 6k+1$,
        we conclude that $Q^{6k+1}$  occurs in $F$ and $G$ at positions $u_i$ and $v_i$, respectively.
        This contradicts the assumption that $P$ avoids vertical periodicity.
    \end{claimproof}

    As discussed above, \cref{clm:vertical_aperiodic_reduction} combined with the original arguments in the proof of  \cite[Lemma 3.12]{DGHKS23}, let us derive $\ted^w(F',G')\le \ted^w(F,G)$.
    The converse inequality follows by symmetry between $(F,G,P)$ and $(F',G',P')$.
\end{proof}

\SetKwFunction{PeriodicityReduction}{PeriodicityReduction}

\begin{fact}[{$\protect\PeriodicityReduction(P,e,  \Qf)$~\cite[Lemma 2.13]{DGHKS23}}]\label{lem:perred}
    Let $e\in \Zp$ and let $\Qf$ be a family of primitive strings of length at most $e$.
    There is an algorithm that repeatedly transforms an input string $P$
    by replacing an occurrence of $Q^{e+1}$ (for some $Q\in \Qf$) with an occurrence of $Q^e$,
    arriving at a string $P'$ that does not contain any occurrence of $Q^{e+1}$ (for any $Q\in \Qf$).
    Moreover, this algorithm can be implemented in linear time using a constant-time oracle that tests whether a given primitive fragment of $P$ belongs to $\Qf$.
\end{fact}

\lemcontextreduction*

\newcommand{\Pp}{\mathbf{P}}
\newcommand{\Qq}{\mathbf{Q}}
\SetKwFunction{ContextReduction}{ContextReduction}
\begin{algorithm}\label{alg:context_reduction}
    \caption{$\protect\ContextReduction(P,k)$}
    \KwIn{a context $P$}
    \KwOut{a context $P'$ with at most $1152k^3$ red characters}
        Let $P=P_0\star \cdots \star P_{e-1}$, where each $P_i$ is a context of depth $1$\;
        $\I \gets \emptyset$\;

        \For{$i\gets 0$ \KwSty{to} $e$}{
            Let $P_i = \langle \op_{a_i} F_i; G_i \cl_{a_i}\rangle$\;
            $F_i' \gets \ForestReduction(F_i,k)$\;
            $G_i' \gets \ForestReduction(G_i,k)$\;
            $\Pp \gets \Pp\cdot \langle \op_{a_i} F_i'; G_i' \cl_{a_i}\rangle$\;
        }
        $\Qf \gets \{\Qq : \bigstar_{i=0}^{|\Qq|-1}\Qq[i]\text{ is a primitive context of length at most }8k\}$\;
        $\Pp' \gets \PeriodicityReduction(\Pp, 6k, \Qf)$\;        
        $P' \gets \bigstar_{i=0}^{|\Pp'|-1}\Pp'[i]$\;
        \If{$|\red{P'}| \le 1152k^3$}{
            \KwRet{$P'$}
        }
        $P'' \gets \bigstar_{i=0}^{24k-1} \context{\op_a^{i+1} \cl_a^{i}}{\op_a^{24k^2-i-1} \cl_a^{24k^2-i}}$ for some $a\in \Sigma$\;
        \KwRet{$P''$}
\end{algorithm}
\begin{proof}
    Let $P=P_0\star \cdots \star P_{e-1}$, where each $P_i$ is a context of depth $1$,
    that is, $P_i = \langle \op_{a_i} F_i; G_i \cl_{a_i}\rangle$ for some label $a_i\in \Sigma$ and forests $F_i,G_i$.
    As the first step, our algorithm constructs a string $\Pp[0\dd e)$ whose characters are depth-1 contexts
    $\Pp[i]=\langle \op_{a_i} F'_i; G'_i \cl_{a_i}\rangle$, where forests $F'_i=\ForestReduction(F_i,k)$ and $G'_i=\ForestReduction(G_i,k)$ are constructed using \cref{alg:forest_reduction}.
    Next, we transform $\Pp$ using \cref{lem:perred} with $e=6k$ and a family $\Qf$ defined so that $\mathbf{P}[i\dd j)\in \Qf$
    if and only if $\Pp[i]\star \cdots \star \Pp[j-1]$ is a primitive context of length at most $8k$ (this implies $j-i\le 4k$).
    In order to apply \cref{lem:perred} to $\Pp$, we use linear-time string sorting~\cite{PT87,AN94} to map characters of $\Pp$ (depth-1 contexts) to integer identifiers.
    By composing the contexts corresponding to the resulting string $\Pp'$, we obtain a context $P'$.
    We return $P'':=\bigstar_{i=0}^{24k-1} \context{\op_a^{i+1} \cl_a^{i}}{\op_a^{24k^2-i-1} \cl_a^{24k^2-i}}$ (for an arbitrary label $a\in \Sigma$)
    or $P'$ depending on whether $|P'|\ge 1152k^3$ or not.

    Note that $|P''| = \sum_{i=0}^{24k-1} (1 + 2\cdot i + 2\cdot (24k^2-i-1)+1) =24k \cdot 2\cdot  24k^2 = 1152k^3$.
    Thus, the resulting context (either $P'$ or $P''$) is guaranteed to be of length at most $1152k^3$.
    Let us now argue that it is $\ted_{\le k}^w$-equivalent to $P$ for every normalized quasimetric~$w$.
    By \cref{lem:horizontal_aperiodic_reduction}, the forests $F'_i$ and $G'_i$ are $\ted_{\le k}^w$-equivalent to $F_i$ and $G_i$,
    respectively, and thus $\bigstar_{i=0}^{e-1} \Pp[i]$ is $\ted_{\le k}^w$-equivalent to $P$.
    By \cref{lem:perred}, the context $P'$ is obtained from $\bigstar_{i=0}^{e-1} \Pp[i]$ by repeatedly replacing $Q^{6k+1}$ with $Q^{6k}$
    for primitive contexts  $Q$ of length at most $8k$.
    By \cref{lem:vertical_periodic_reduction}, $Q^{6k+1}$ is then  $\ted_{\le k}^w$-equivalent to $Q^{6k}$,
    so this operation preserves $\ted_{\le k}^w$-equivalence, i.e., $P'$ is also $\ted_{\le k}^w$-equivalent to $P$.
    Moreover, each depth-$1$ context in $\Pp'$ originates from $\Pp$, so each forest occurring in (either half of) $P'$ is of length at most $158k^2$.
    Furthermore, \cref{lem:perred} guarantees that $P'$ is not of the form $C\star Q^{6k+1} \star D$ for any context $Q$ of length at most $8k$,
    and thus $P'$ avoids vertical $k$-periodicity.
    By construction, $P''$ avoids vertical $k$-periodicity and its halves contain only forests of lengths at most $158k^2$ (in fact, at most $48k^2$).
    Also, $P''$ does not contain any black character, since it does not have any periodic block.
    Thus, $|\red{P''}|=|P''|=1152k^3$.
    Consequently, \cref{lem:context_reduction} implies that $P''$ is $\ted_{\le k}^w$-equivalent to $P'$ (and, by transitivity, to $P$)
    provided that  $|P''|\ge 1152k^3$.

    As for the running time analysis, we note that all applications of \cref{lem:horizontal_aperiodic_reduction}
    concern disjoint fragments of $P$, so the total cost of the calls to $\ForestReduction$ is linear.
    Assigning integer identifiers to contexts $\Pp[i]$ and applying \cref{lem:perred} also takes linear time.
    Finally, $P''$ is constructed only if $|P'|\ge 1152k^3$, so the cost of this step is also be bounded by $\Oh(|P|)$.
\end{proof}

\subsection{Complete Kernelization Algorithm}\label{app:kernel}
In these section, we combine the results proved above to derive \cref{thm:our_kernel}.
Notably, \cref{lem:alg:forest_reduction,lem:alg:context_reduction} are formulated in terms of red characters (originating from periodic blocks) whereas \cref{thm:our_kernel} requires constructing free block.
Recall that the string period $R$ of a free block must be balanced, whereas the string period $Q$ of a periodic block only needs to have the same number of opening and closing parentheses; see \cref{def:free,def:bf}.
We use the following sequence of results to conclude that, as long as the periodic block is contained in a forest, some cyclic rotation of $Q$ must be balanced.

Simple inductive arguments prove the following characterizations.
\begin{observation}\label{rmk:forest}
   Every non-empty balanced string $F\in \FSigma$ can be expressed as $F=\op_a \cdot F' \cdot \cl_a \cdot F''$ for some $F',F''\in \FSigma$ and $a\in \Sigma$.
\end{observation}

\begin{observation}\label{obs:balanced}
 Every balanced string $F\in \FSigma$ satisfies the following conditions:
 \begin{enumerate}
    \item $F$ contains equally many opening and closing parentheses, and
    \item every prefix of $F$ contains at least as many opening as closing parentheses.
 \end{enumerate}
\end{observation}

We write $\FsSigma \subseteq \PSigma^*$ for the family of substrings of balanced strings in $\FSigma$.
In other words, $F\in \FsSigma$ if and only if $F=G[i\dd j)$ for some balanced string $G\in \FSigma$ and integers $0\le i \le j \le |G|$.
The converse of \cref{obs:balanced} remains true among the elements of~$\FsSigma$.

\begin{fact}\label{obs:balanced_fragment}
    If a string $F\in \FsSigma$ satisfies the following conditions, then it belongs to $\FSigma$:
    \begin{enumerate}
        \item $F$ contains equally many opening and closing parentheses, and
        \item every prefix of $F$ contains at least as many opening as closing parentheses.
    \end{enumerate}
\end{fact}
\begin{proof}
    We proceed by induction on the length $|G|$ of a balanced superstring $G\in \FSigma$ of $F$.
    In the base case, we have $F=G$, and the claim holds trivially.
    Otherwise, $F$ has an occurrence in $G$ that does not contain the first or the last character of $G$.
    The two cases are symmetric, so we henceforth assume that the occurrence does not contain the first character of $G$.
    By \cref{rmk:forest}, we have $G = \op_a \cdot G' \cdot \cl_a \cdot G''$ for some $G',G''\in \FSigma$ and $a\in \Sigma$. 
    Moreover, as argued above, we can assume that $F$ is a substring of $G' \cdot \cl_a \cdot G''$.
    If $F$ is a substring of $G'$ or $G''$, then it is balanced by the inductive hypothesis.
    In the remaining case, $F=F'\cdot \cl_a \cdot F''$, where $F'$ is a suffix of $G'$ and $F''$ is a prefix of $G''$.
    By \cref{obs:balanced}, the suffix $F'$ of $G'\in \FSigma$ has at least as many closing parentheses as opening parentheses.
    Hence, the prefix $F'\cdot \cl_a$ of $F$ has strictly more closing parentheses than opening parentheses.
    This contradiction concludes the proof.
\end{proof}

\begin{lemma} \label{lem:balanced_cyclic_rotation}
    If $Q\in\PSigma^*$ has equally many opening and closing parentheses and $Q^2\in \FsSigma$, then $Q$ has a balanced cyclic rotation, i.e., $Q[i\dd |Q|)Q[0\dd i)\in \FSigma$ for some $i\in [0\dd |Q|]$.
\end{lemma}
\begin{proof}
    For every $i\in [0\dd |Q|]$, let $s_i$ denote the number of opening parentheses in $Q[0\dd i)$ minus the number of closing parentheses in $Q[0\dd i)$.
    Note that $s_{|Q|}=0$ because $Q$ has equally many  opening and closing parentheses.
    Let $p\in [0\dd |Q|]$ be an index with minimum value of $s_p$.
    We claim that $Q' = Q[p\dd |Q|)Q[0\dd p)$ is a balanced cyclic rotation of $Q$. 

    As a cyclic rotation of $Q$, the string $Q'$ has equally many  opening and closing parentheses.
    Moreover, $Q'$ is a substring of $Q^2$ and, by transitivity, belongs to $\FsSigma$.
    Thus, by \cref{obs:balanced_fragment}, it suffices to show that, for each $i\in [0\dd |Q|]$, we have $s'_i\ge 0$, where $s'_i$ is the number of opening parentheses in $Q'[0\dd i)$ minus the number of closing parentheses in $Q'[0\dd i)$.
    \begin{itemize}
        \item If $i \le |Q|-p$, then $s'_i = s_{p+i} - s_p \ge 0$ since $s_p \le s_{i+p}$ holds by definition of $p$.
        \item If $i > |Q|-p$, then $s'_i = s_{|Q|} - s_p + s_{i-(|Q|-p)} = s_{i-(|Q|-p)} - s_p \ge 0$ since $s_{|Q|}=0$ and $s_p \le s_{i-(|Q|-p)}$ holds by definition of $p$.\qedhere
    \end{itemize}
\end{proof}

We are now ready to prove \cref{thm:our_kernel}, which we restate below for convenience.

\thmkernel*

\begin{proof}
    By \cref{thm:prev_kernel}, we can assume without loss of generality that the forests $F$ and $G$ are already of size $\Oh(k^5)$.
    We start by applying algorithm of Theorem~\ref{thm:decomposition} on forests $F$ and $G$ to get a matching of size $\Oh(k)$.
    Next, for each pair of the matching that is a forest $P$, we apply the algorithm of Lemma~\ref{lem:horizontal_aperiodic_reduction} to get a $\ted_{\le k}^w$-equivalent forest $P'$ with $|\red{P'}| \le 158k^2$, and replace that occurrence of $P$ in both $F$ and $G$ with $P'$.
    Similarly, for each pair of the matching that is a context $P=\context{P_L}{P_R}$, we apply the algorithm of \cref{lem:context_reduction} to get a $\ted_{\le k}^w$-equivalent context $P'$ with $|\red{P'}| \le 1152k^3$, and replace that occurrence of $P$ in both $F$ and $G$ with $P'$.
    After applying these reduction algorithms on the matching pairs, we can apply Lemma~\ref{lem:findgood} and \cref{lem:findgood_context}
    on each $P'$ to obtain sets of periodic blocks, each of size $\Oh(k^2)$.
    For each of these periodic blocks of matching pieces,
    let their occurrence in $F'$ and $G'$ be at $F'[l_F\dd r_F)$ and $G'[l_G\dd r_G)$, respectively.
    We can construct a free pair from this periodic block as follows:
    \begin{itemize}
        \item Let $Q$ be the shortest string period of this block. Since $Q$ has equally many opening and closes parentheses, \cref{lem:balanced_cyclic_rotation} shows that it has a balanced cyclic rotation $Q'$.
        We remove at most $|Q|$ characters from the beginning and end of the periodic block to get a power of $Q'$. Let the new periodic block occur at $F'[l_F'\dd r_F')$ and $G'[l_G'\dd r_G')$.
        \item Let $R=Q'^{\lceil \frac{4k}{|Q'|} \rceil}$ be the minimum power of $Q'$ such that $|R|\ge 4k$.
        We can observe that $|R|\le 8k$ since $|Q'|\le 4k$.
        \item Let $r_F''=l_F'+\lfloor \frac{r_F'-l_F'}{|R|}\rfloor\cdot|R|$ and $r_G''=l_G'+\lfloor \frac{r_G'-l_G'}{|R|}\rfloor\cdot|R|$.
        \item The free pair is then $F'[l_F'+|R|\dd r_F''-|R|)$ and $G'[l_G'+|R|\dd r_G''-|R|)$.
    \end{itemize}
    The number of these free pairs is $\Oh(k^3)$, since there are $\Oh(k)$ matching pieces and each matching piece has $\Oh(k^2)$ periodic blocks.
    The number of red characters of each $P'$ is bounded by $\Oh(k^3)$, so the total number of red characters of these matching pieces if $\Oh(k^4)$.
    Each periodic block ignores at most $3|R|+2|Q|\le 32k$ characters of the periodic block, and there are at most $\Oh(k^3)$ such periodic blocks.
    Also, there are at most $\Oh(k^4)$ red characters from the matching pairs, and $\Oh(k^4)$ non-matched characters.
    Therefore, the created free pairs of $F'$ and $G'$ leave at most $\Oh(k^4)$ non-free characters.
\end{proof}

%% file: src/forest_decomposition.tex
The starting point of our algorithm for~\cref{thm:decomposition} is the following decomposition of~\cite{DGHKS23}.

\begin{definition}[{\cite[Definition 3.14]{DGHKS23}}]\label{def:decomp}
    A set $\D\sub \Pc(F[i\dd j))$ is a \emph{piece decomposition} of 
    a balanced fragment $F[i\dd j)$ of a forest $F$ if it satisfies one of the following conditions:
    \begin{itemize}
        \item $\D=\emptyset$ and $i=j$;
        \item $\D = \{F[i\dd j)\}$ and $i<j$;
        \item $\D = \D_L\cup \D_R$ for some piece decompositions $\D_L$ of $F[i\dd m)$ and $\D_R$ of $F[m\dd j)$,
        where $m\in (i\dd j)$.
        \item $\D = \{\langle F[i\dd i');F[j'\dd j)\rangle\}\cup \D'$ for a context $\langle F[i\dd i');F[j'\dd j)\rangle\in \Pc(F)$
        and a piece decomposition $\D'$ of $F[i'\dd j')$.
    \end{itemize}
    \end{definition}

\begin{lemma}[{\cite[Lemma~3.15]{DGHKS23}}]\label{lem:decomp}
    There exists a linear-time algorithm that, given a forest $F$ and an integer $t\ge 2$,
    constructs a piece decomposition of $F$ consisting of at most $\max(1,\frac{6|F|}{t}-1)$ pieces of length at most $t$ each.\footnote{Each piece consists of at most $t$ characters, corresponding to at most $t/2$ nodes.}
\end{lemma}

\begin{definition}
    For two forests $F$ and $G$, and a piece decomposition $\D$ of $F$, a \emph{piece matching} between $\D$ and $G$ is a matching $\M$ between $F$ and $G$ such that $\M \subseteq \D \times \Pc(G)$.
\end{definition}

The following lemma shows that, for any piece decomposition $\D$ of $F$, there exists a good-enough piece matching between $\D$ and $G$.

\begin{lemma}\label{lem:good_matching_exists}
    Let $F$ and $G$ be forests with $\ted(F,G)\le k$ and let $\DC$ be a piece decomposition of $F$.
    There exists a piece matching $\M$ between $\DC$ and $G$ such that $\M$ leaves at most $2k$ pieces of $\DC$ and at most $4k$ fragments of $G$ unmatched.
\end{lemma}
\begin{proof}
    Since $\ted(F,G)\le k$, then there is an optimal alignment $\A \in \bta{k}(F,G)$ of $F$ and $G$ with at most $2k$ unmatched positions.  
    Consider a decomposition of (the string) $F=F_1\cdots F_{p}$ into at most $p\le 2|\DC|$ fragments such that each piece in $\DC$ corresponds to one or two fragments (depending on whether it is a subforest or a context).
    The alignment $\A$ yields a decomposition of (the string) $G=G_1\cdots G_p$ so that, for each $i\in [1\dd p]$, the alignment $\A$ aligns $F_i$ to $G_i$.
    For each piece $d$ of $\DC$, if $\A$ does not perfectly match the corresponding fragment(s) $F_i$ to $G_i$, we can skip $d$.
    For all other pieces, we can find a piece in $G$ that is the image of that piece under $\A$.
    The resulting pairs of pieces form a set satisfying \cref{def:matching}.
    Moreover:
    \begin{itemize}
        \item 
        The number of unmatched pieces of $\DC$ is at most $2k$ because, for each edit in $\A$, we skip at most one piece of $\DC$.
        \item
        At most $4k$ fragments $F_i$ correspond to unmatched pieces of $\DC$, and for each of these fragments, the corresponding fragment $G_i$ is also unmatched.
        Hence, the number of unmatched fragments in $G$ is at most $4k$.\qedhere
    \end{itemize}
\end{proof}

Now, we want to find a piece matching $\M$ with at most $\Oh(k)$ unmatched pieces in $\DC$ and $\Oh(k)$ unmatched fragments in $G$. 
By the previous lemma, we know that such a matching exists.
To do this, we find the matching that minimizes the number of unmatched pieces in $\DC$ plus the number of unmatched fragments in $G$.

For a piece decomposition $\DC$ of $F$ and a balanced fragment $F[i\dd j)$ such that each piece of $\DC$ is either fully contained in $F[i\dd j)$ or disjoint from it,
we denote by $\D_{i,j}$ the piece decomposition of $F[i\dd j)$ induced by $\DC$.

\begin{definition}
    Consider a piece decomposition $\DC$ of a forest $F$ and an arbitrary fragment $G[i'\dd j')$ of a forest $G$.
    Let $\M$ be a piece matching between $\DC_{i,j}$ and $G$ such that every piece in $\D_{i,j}$ is unmatched or matched with a piece contained in $G[i'\dd j')$.
    We define a cost function $\cost(\D_{i,j},G[i'\dd j'), \M)$ as the number of unmatched pieces of $\DC_{i,j}$ plus the number of unmatched fragments of $G[i'\dd j')$ when matching $\D_{i,j}$ to $G[i'\dd j')$ using $\M$. 
\end{definition}
\begin{lemma}\label{lem:dp}
    Given forests $F$ and $G$ of total size $n$, a threshold $k\in \Zz$, and a piece decomposition $\DC$ of $F$,
    one can find in $\Oh(n+|\DC|k^3)$ time a matching $\M$ between the pieces of $\DC$ and $G$ that minimizes 
    $\cost(\D_{0,|F|},G[0\dd |G|),\M)$.
\end{lemma}
\SetKwFunction{Pair}{Pairs}
\newcommand{\Ss}{\mathcal{S}}
\newcommand{\false}{\texttt{0}\xspace}
\newcommand{\true}{\texttt{1}\xspace}
\begin{algorithm}
    \SetInd{0.3em}{0.3em}
    \caption{Compute an optimal matching $\M$ of $D_{i,j}$ to $G[i'\dd j')$ with respect to $\cost$.}\label{alg:dp}
    \Fn{$\Pair(\D_{i,j},G[i'\dd j'), f_l, f_r)$}{
        $result \gets \infty$\;
        \If{$f_l=\false$ \KwSty{and} $f_r=\false$}{\label{alg:dp:case_no_match}
            $result \mgets |\D_{i,j}| + 1$\;
        }
        \If{$i' < \min(j'-1,i+2k)$ \KwSty{and} $f_l=\false$}{\label{alg:dp:case_unmatched_left}
            $result \mgets \Pair(D_{i,j}, G[i'+1\dd j'), \false, f_r)$\;
            $result \mgets \Pair(D_{i,j}, G[i'+1\dd j'), \true, f_r) + 1$\;
        }
        \If{$j' > \max(i'+1,j-2k)$ \KwSty{and} $f_r=\false$}{\label{alg:dp:case_unmatched_right}
            $result \mgets \Pair(D_{i,j}, G[i'\dd j'-1), f_l, \false)$\;
            $result \mgets \Pair(D_{i,j}, G[i'\dd j'-1), f_l, \true) + 1$\;
        }
        \If{$\D_{i,j}=\{F[i\dd j)\}$ \KwSty{and} $F[i\dd j)=G[i'\dd j')$ \KwSty{ and } $f_l=f_r=\true$}{\label{alg:dp:case_full_match}
            $result \mgets 0$\;
        }
        \If{$\D_{i,j}=\D_{i,m}\cup \D_{m,j}$ for some $m\in (i\dd j)$}{\label{alg:dp:case_divide}
            \ForEach{$m'\in [m-2k\dd m+2k]\cap [i'\dd j']$}{
                \If{$m' = i'$}{\label{alg:dp:case_divide:empty_left}
                    $result \mgets  |\D_{i, m}| + \Pair(\D_{m, j}, G[i'\dd j'), f_l, f_r)$\;
                }
                \ElseIf{$m' = j'$}{\label{alg:dp:case_divide:empty_right}
                    $result \mgets \Pair(\D_{i, m}, G[i'\dd j'), f_l, f_r) + |\D_{m, j}|$\;
                }
                \Else{\label{alg:dp:case_divide:general}
                    $result \mgets \Pair(\D_{i,m}, G[i'\dd m'), f_l, \false) + \Pair(\D_{m,j}, G[m'\dd j'), \false, f_r) - 1$\;\label{alg:dp:case_divide:general:both_flags_false}
                    $result \mgets \Pair(\D_{i,m}, G[i'\dd m'), f_l, \false) + \Pair(\D_{m,j}, G[m'\dd j'), \true, f_r)$\;
                    $result \mgets \Pair(\D_{i,m}, G[i'\dd m'), f_l, \true) + \Pair(\D_{m,j}, G[m'\dd j'), \false, f_r)$\;
                    $result \mgets \Pair(\D_{i,m}, G[i'\dd m'), f_l, \true) + \Pair(\D_{m,j}, G[m'\dd j'), \true, f_r)$\;
                }
            }
        }
        \If{$\D_{i,j}=\{\langle F[i\dd i+\ell); F[j-r\dd j)\rangle\}\cup \D_{i+\ell,j-r}$}{\label{alg:dp:case_context}
            \If{$F[i\dd i+\ell)=G[i'\dd i'+\ell)$ \KwSty{and} $F[j-r\dd j)=G[j'-r\dd j')$ \KwSty{and} $G[i'+\ell\dd j'-r)$ is balanced \KwSty{and} $f_l=f_r=\true$}{\label{alg:dp:case_context:match}
                \lIf{$i'+\ell=j'-r$}{
                    $result \mgets |\D_{i+\ell,j-r}|$
                }
                \Else{%
                    \ForEach{$(f_l', f_r') \in \{\false, \true\}\times \{\false,\true\}$}{
                        $result \mgets \Pair(\D_{i+\ell,j-r}, G[i'+\ell\dd j'-r), f_l', f_r')$\;
                    }
                }
            }
            $i'' \gets \max(i+\ell-2k,i')$; \ \ $j'' \gets \min(j-r+2k,j')$\;\label{alg:dp:case_context:unmatched}
            $S_l \gets \{f_l\}$; \ \ $S_r \gets \{f_r\}$\;
            \If{$i' < i''$}{
                \lIf{$f_l=\false$}{$S_l \gets \{\false, \true\}$}
                \lElse{$S_l \gets \varnothing$}
            }
            \If{$j'' < j'$}{
                \lIf{$f_r=\false$}{$S_r \gets \{\false, \true\}$}
                \lElse{$S_r \gets \varnothing$}
            }
            \If{$i'' < j''$}{
                \ForEach{$(f_l',f_r') \in S_l\times S_r$}{
                        $result \mgets \Pair(\D_{i+\ell,j-r}, G[i''\dd j''), f_l', f_r') + \mathbf{1}[f_l <
                         f_l'] + \mathbf{1}[f_r < f_r'] + 1$\;
                }
            }
        }
        \KwRet{$result$} \\
    }
\end{algorithm}
\begin{proof}
    We modify \cite[Algorithm 5]{DGHKS23} to compute the new cost function.
    We define the dynamic programming to compute the minimum cost matching between a piece decomposition $\D_{i,j}$ and a non-empty fragment $G[i'\dd j')$.
    We also add two boolean flags $f_l$ and $f_r$ to indicate whether the leftmost and rightmost characters of $G[i'\dd j')$ are matched or not.
    If either of the flags is \false, then the corresponding endpoint is unmatched; otherwise, it is matched.

    \Cref{alg:dp} describes our dynamic programming algorithm as a recursive procedure \Pair($\D_{i,j},G[i'\dd j'), f_l, f_r$) that computes the minimum cost of a matching between the piece decomposition $\D_{i,j}$ of $F$ and the fragment $G[i'\dd j')$,
    with the flags forcing the leftmost and rightmost characters of $G[i'\dd j')$ to be matched or unmatched.
    We also assume $i' \in [i-2k\dd i+2k]$ and $j' \in [j-2k\dd j+2k]$.
    The algorithm consists of the following cases:
    \begin{enumerate}
        \item\label{it:case:empty} If both flags are set to \false, we can leave all pieces unmatched and return the number of pieces in $\D_{i,j}$, plus one for the unmatched fragment $G[i'\dd j')$.
        This case is covered in line \ref{alg:dp:case_no_match}.
        \item\label{it:case:triml} If $f_l$ is set to \false and $i' < \min(j'-1,i+2k)$, then we can leave $G[i']$ unmatched and match $G[i'+1\dd j')$ to $\D_{i,j}$.
        This case is covered in line~\ref{alg:dp:case_unmatched_left}.
        \item\label{it:case:trimr} If $f_r$ is set to \false and $j' > \max(i'+1,j-2k)$, then we can leave $G[j'-1]$ unmatched and match $G[i'\dd j'-1)$ to $\D_{i,j}$.
        This case is covered in line~\ref{alg:dp:case_unmatched_right}.
        \item\label{it:case:single} If $\D_{i,j}$ consists of a single forest $F[i\dd j)$ identical to $G[i'\dd j')$, and both flags are set to \true, then we can match $F[i\dd j)$ to $G[i'\dd j')$ and return zero.
        This case is covered in line~\ref{alg:dp:case_full_match}.
        \item\label{it:case:split} If $\D_{i,j}$ consists of two separate decompositions $\D_{i,m}$ and $\D_{m,j}$ split at some position $m$, then we can fix a position $m'$ by a distance of at most $2k$ from $m$, split $G[i'\dd j')$ at that position,
        and then recursively match $\D_{i,m}$ with $G[i'\dd m')$ and $\D_{m,j}$ with $G[m'\dd j')$.
        This case is covered in line~\ref{alg:dp:case_divide}.
        Since our recursion disallows empty fragments of $G$, we have two special cases when $G[i'\dd m')$ or $G[m'\dd j')$ is empty; those cases are covered in lines~\ref{alg:dp:case_divide:empty_left} and \ref{alg:dp:case_divide:empty_right}.
        The general case where both pieces are non-empty is covered in line~\ref{alg:dp:case_divide:general}.
        Here, we iterate over the choices of the right flag for the left piece and the left flag for the right piece.
        If both new flags are \false, then the unmatched suffix of $G[i'\dd m')$ gets merged with the unmatched prefix of $G[m'\dd j')$, so we subtract 1 from the sum of the two costs, as seen in line~\ref{alg:dp:case_divide:general:both_flags_false}.
        \item\label{it:case:context} If the previous cases do not apply, then $\D_{i,j}$ is wrapped in a context $\context{F[i\dd i+\ell)}{\allowbreak F[j-r\dd j)}$.
        In this case, we either match the context or skip it.
        If $\context{G[i'\dd j'+\ell)}{G[j'-r\dd j')}$ is a matching context in $G$ and both flags are set to \true, then we can match the context and recursively matching $\D_{i+\ell,j-r}$ with $G[i'+\ell\dd j'-r)$.
        This case is covered in line~\ref{alg:dp:case_context:match}.
        In case we decide not to match the context, we should first trim $G[i'\dd j')$ from both endpoints to $G[i''\dd j'')$  so that the new endpoints are of distance at most $2k$ from the new endpoints of the piece decomposition $\D_{i+\ell,j-r}$.
        This case is covered from line~\ref{alg:dp:case_context:unmatched}.
        Here, we define two sets $S_l$ and $S_r$ which contain the possible values for the new left and right flags, denoted $f_l'$ and $f_r'$.
        If any of the endpoints of $G'$ is to be trimmed, then the corresponding original flag should have been set to \false, whereas the new flag can have any value.
        After trimming, we should match $D_{i+\ell,j-r}$ to $G[i''\dd j'')$,
        and we should pay 1 for the unmatched piece, plus 1 for each endpoint with a new unmatched fragment, i.e., with the original flag set to \false and the new flag set to \true.
    \end{enumerate}
    The case analysis above proves that the computed value can be realized as the cost of some matching $\M$.
    Next, we prove that any matching $\M$ between $\D_{i,j}$ and $G[i'\dd j')$ can be obtained by one of the cases above.
    This can be achieved by induction on $(j-i) + (j'-i')$.
    \begin{enumerate}[label=(\alph*)]
        \item If $\M = \emptyset$, then it is covered in case \ref{it:case:empty}. We henceforth assume $\M \ne \emptyset$.
        \item If $|\D_{i,j}|=1$, then $\M$ matches $F[i\dd j)$---the only piece of $\D_{i,j}$---to a fragment $G[i''\dd j'')$ contained in $G[i'\dd j')$. This matching is created by $i''-i'$ applications of case \ref{it:case:triml}, $j'-j''$ applications of case \ref{it:case:trimr}, and finally an application of 
        case \ref{it:case:single}.
        Note that $|i-i''|\le 2k$ and $|j-j''|\le 2k$, so the applications of cases~\ref{it:case:triml}--\ref{it:case:trimr} are valid.
        \item If $\D$ contains a context $\context{F[i\dd i+\ell)}{F[j-r\dd j)}$, then $\M$ may match the context or not.
        If $\M$ matches the context to $\context{G[i''\dd i''+\ell)}{G[j''-r\dd j'')}$, then this is handled by $i''-i'$ applications of case \ref{it:case:triml}, $j'-j''$ applications of case \ref{it:case:trimr}, and finally an application of case \ref{it:case:context}.
        Note that $|i-i''|\le 2k$ and $|j-j''|\le 2k$, so the applications of cases~\ref{it:case:triml}--\ref{it:case:trimr} are valid.
        If $\M$ does not match $\context{F[i\dd i+\ell)}{F[j-r\dd j)}$, then it may only match the remaining pieces of $\D_{i,j}$ to pieces of $G$ contained within $G[\max(i',i+\ell-2k)\dd \min(j',j-r+2k))$. 
        Thus, $\M$ is handled directly in case \ref{it:case:context}.
        \item Otherwise, $\D_{i,j}$ can be decomposed as $\D_{i,j}=\D_{i,m}\cup \D_{m,j}$.
        In that case, $\M$ can be represented as a union of two matchings $\M_L$ and $\M_R$ from $\D_{i,m}$ and $\D_{m,j}$ to disjoint fragments of $G[i'\dd j')$.
        Since all pieced of $\M_L$ are to the left of elements of $\M_R$, projected on $\D$, then their image on $G$ satisfies the same property.
        That is, $G[i'\dd j')$ can be split to two fragments $G[i'\dd m')$ and $G[m'\dd j')$ such that $\M_L$ is a matching between $\D_{i,m}$ and $G[i'\dd m')$, and $\M_R$ is a matching between $\D_{m,j}$ and $G[m'\dd j')$.
        Also, since the width of the matching is at most $2k$, then the position $m'$ of the split can be chosen within $[m-2k\dd m+2k]$.
        Therefore, this $\M$ can be created in case \ref{it:case:split}, where we fix the position of $m'$ and recursively compute $\M_L$ and $\M_R$.
    \end{enumerate}

    As for the running time of the algorithm, a simple induction shows that at most $|\D_{i,j}|$ different piece decompositions are visited in the recursion tree.
    For each unique $\D_{i,j}$, there are $\Oh(k^2)$ possible fragments $G[i'\dd j')$ and at most $4$ different values of $f_l$ and $f_r$.
    Therefore, the number of unique calls to \Pair is $\Oh(|\D|k^2)$.
    For the efficient implementation of the algorithm, we can use a memoization table to store the results of the recursive calls.
    Each call performs $\Oh(k)$ instructions and each instruction can be implemented in $\Oh(1)$ time after linear-time preprocessing~\cite[Theorem 2.1 and Fact 3.1]{DGHKS23}, needed to efficiently check whether fragments of $F$ match fragments of $G$ and whether fragments of $F$ and $G$ are balanced.
    Including the necessary preprocessing, the overall runtime is $\Oh(n+|\D|k^3)$.
\end{proof}

Let us first run the algorithm of Lemma~\ref{lem:decomp} with $t=k^3$ to get a piece decomposition $\DC$ of $F$ with $\Oh(\frac{|F|}{k^3})$ pieces of length at most $k^3$ each.
Then we can run the algorithm of Lemma~\ref{lem:dp} on $\DC$ and $G$ in $\Oh(n+\frac{|F|}{k^3}k^3) = \Oh(n)$ time to get a minimum-cost piece matching between $\DC$ and $G$.
By \cref{lem:good_matching_exists}, the cost of this matching does not exceed $6k$.
In particular, there are at most $6k$ unmatched pieces of $\DC$, and the unmatched characters of $G$ form at most $6k$ fragments.
Each piece of $\DC$ is of length at most $k^3$, so we have at most $6k^4$ unmatched characters of $F$.
But we are not done yet, because there are too many pieces in the matching, and we need to reduce their number to $\Oh(k)$.

\begin{definition}
    Consider forests $F$ and $G$ and a matching $\M$ of $F$ and $G$.
    We say that a piece $f\in \Pc(F)$ is \emph{clean} with respect to $\M$ if there is a piece $g\in \Pc(G)$ with $(f,g)\in \M$,
    there is no unmatched character in $F$ adjacent to $f$,
    and there is no unmatched character in $G$ that is adjacent to $g$.
    If a piece of $F$ is not clean, we call it \emph{dirty} with respect to $\M$.
\end{definition}

Our strategy is to repeatedly merge \emph{adjacent} clean pieces of the decomposition $\DC$.

\begin{definition}\label{def:adjacent}
Two pieces $f,f'$ of a forest $F$ are \emph{adjacent} if one of the following holds:
\begin{enumerate}
    \item $f=F[l\dd r)$ and $f'=F[l'\dd r')$ are both subforests with $r=l'$; in this case, we define $f\cup f'$ to be a subforest $F[l\dd r')$.
    \item $f = \context{F[l_L\dd r_L)}{F[l_R\dd r_R)}$ is a context and $f'=F[l'\dd r')$ is a subforest with $r_L=l'$ or $l_R=r'$; in this case, we define $f\cup f'$ to be a forest $F[l_L\dd r_R)$ (if both $r_L=l'$ and $l_R=r'$), a context $\context{F[l_L\dd r')}{F[l_R\dd r_R)}$ (if just $r_L=l'$), or a context $\context{F[l_L\dd r_L)}{F[l'\dd r_R)}$ (if just $l_R=r'$).
    \item $f = \context{F[l_L\dd r_L)}{F[l_R\dd r_R)}$ and $f' = \context{F[l'_L\dd r'_L)}{F[l'_R\dd r'_R)}$ are both contexts with $r_L=l'_L$ and $l_R=r'_R$;
    in this case, we define $f\cup f'$ to be a context $\context{F[l_L\dd r'_L)}{F[l'_R\dd r_R)}$.
\end{enumerate}
\end{definition}

We represent $\DC$ using the following auxiliary forest $\HC_\DC$ so that adjacent clean pieces can be identified efficiently.

\begin{definition}
    For a piece decomposition $\DC$ of a forest $F$, we define a forest $\HC_\DC$, called the \emph{hierarchy of $\DC$}, by collapsing each piece of $\DC$ to a single node whose label indicates the underlying piece.
    In other words, each subforest $f=F[l\dd r)\in \DC$ is collapsed to $\op_f \cl_f$;
    for each context $f = \context{F[l_L\dd r_L)}{F[l_R\dd r_R)}$, we collapse $F[l_L\dd r_L)$ to $\op_f$
    and $F[l_R\dd r_R)$ to $\cl_f$.
\end{definition}

The following simple observation characterizes adjacent pieces and describes the effect of merging two adjacent clean pieces.
\begin{observation}\label{obs:merge}
    Consider a piece decomposition $\DC$ of a forest $F$. Pieces $f,f'\in \DC$ are adjacent if and only if the corresponding nodes $v,v'$ of $\HC_\DC$ satisfy one of the following:
    \begin{enumerate}
        \item $v'$ is the sibling immediately to the right of $v$, and $v,v'$ are both leaves;
        \item $v'$ is the leftmost or the rightmost child of $v$, and $v'$ is a leaf.
        \item $v'$ is the only child of $v$, and $v'$ is \emph{not} a leaf.
    \end{enumerate}
    In all these cases, $\DC':=\DC \setminus\{f,f'\}\cup \{f\cup f'\}$ is a piece decomposition of $F$.

    Additionally, consider a piece matching $\M$ between $\DC$ and a forest $G$.
    If $f,f'$ are both clean with respect to $\M$, then there exist adjacent pieces $g,g'\in \Pc(G)$ such that $(f,g),(f',g')\in \M$.
    Moreover, $\M':=\M\setminus\{(f,g),(f',g')\} \cup \{(f\cup f', g\cup g')\}$ is a piece matching between $\DC'$ and~$G$.
    Furthermore, $f\cup f'$ is clean with respect to $\M'$, and every piece in $\DC \setminus\{f,f'\}$ that is clean with respect to $\M$ remains clean with respect to $\M'$.
\end{observation}

\begin{lemma}\label{lem:merge}
    Consider forests $F$ and $G$, a piece decomposition $\DC$ of $F$, and a piece matching $\M$ between $\DC$ and $G$.
    If $\DC$ contains $t$ dirty pieces with respect to $\M$, then one can repeatedly apply the merges described in \cref{obs:merge} to derive a piece decomposition $\DC'$ of $F$ with at most $5t+1$ pieces in total and a piece matching $\M'$ between $\DC'$ and $G$.
    The underlying transformation can be implemented in $\Oh(n)$ time.
\end{lemma} \begin{proof}
    The algorithm gradually transforms $\DC$ and $\M$ while performing a recursive traversal of the hierarchy $H_\DC$.
    We maintain the following invariant for every node $v$ of $H_\DC$: if the subtree of $v$ contains $d_v$ dirty nodes, then it is shrunk to at most $\max(1,5d_v-2)$ nodes if $v$ is clean, and at most $5d_v-3$ nodes if $v$ is dirty.
  
    If $v$ is a leaf, we do not perform any action. The invariant is satisfied since the claimed bounds on the subtree of $v$ are all at least $1$.

    Next, suppose that $v_1,\ldots, v_r$ are all the children of a node $v$ or all the roots of $\HC_\DC$.
    Suppose that, for each $i\in [1\dd r]$, the subtree of $v_i$ contains $d_i$ dirty nodes and,
    after processing, has been shrunk to size $s_i \le \max(1, 5d_i-2)$.
    Moreover, let $q$ denote the number of \emph{dirty subtrees} with $d_i>0$, and let $d$ be the total number of dirty nodes in these subtrees.
    After the recursive transformations, each dirty subtree has at most $5d_i-2$ nodes, for a total of at most $5d-2q$ nodes.
    The remaining \emph{clean subtrees} are each of size one. 
    Since \cref{obs:merge} allows merging clean subsequent leaves, we exhaustively merge subsequent clean subtrees; this reduces the number of clean subtrees to $q+1$ and the total number of nodes to $5d-2q+q+1 = 5d-q+1 \le 5d+1$.
    In particular, if $v_1,\ldots, v_r$ are the roots of $\HC_\DC$, then the obtained decomposition consists of at most $5d+1=5t+1$ pieces.

    Otherwise, $v_1,\ldots, v_r$ are the children of a node $v$ representing a context.
    If $v$ is dirty, then the subtree of $v$ is already of size at most $5d+1+1=5d_v-3$.
    When $v$ is clean, we merge it with its clean children as allowed by \cref{obs:merge}.
    We consider three cases:
    \begin{description}
        \item[$q=0$.] After merging clean subtrees, the node $v$ has exactly one child $v'$, which is a leaf. The nodes $v$ and $v'$ are merged, resulting in $1$ node left in the subtree of $v$.
        \item[$q=1$.]  After merging clean subtrees, the node $v$ has exactly one child $v'$ with a dirty subtree and at most two children with clean subtrees (one to the left and one to the right of $v'$). 
        The clean subtrees are merged with $v$ so that $v'$ becomes the only child of $v$. 
        If $v'$ is clean, it is also merged with $v$.
        In this case, the subtree of $v'$ is of size at most $5d-2$ and, after all the merges, the same bound holds for the subtree of $v$.
        If $v'$ is dirty, then its subtree is of size at most $5d-3$, whereas subtree of $v$ is of size at most $5d-2$.
        \item[$q\ge 2$.] In this case, the node $v$ has exactly $q$ children with dirty subtrees. After merging clean subtrees, there are at most $q-1$ clean subtrees in between dirty subtrees. 
        The remaining at most two clean subtrees (corresponding to the leftmost and the rightmost children of $v$) are merged with $v$, so the total number of nodes in the subtree of $v$ does not exceed $5d-2q + (q-1)+1 \le 5d-2q+2 \le 5d-2$. 
    \end{description}
    
    An efficient implementation of the algorithm described and analyzed above requires constructing the hierarchy $\HC_\DC$, partitioning the nodes among clean and dirty, and a post-order traversal of $\HC_\DC$.
    The first two steps are easy to implement in $\Oh(n)$ time, whereas the last one takes $\Oh(|\DC|)=\Oh(n)$ time.
\end{proof}


We are now ready to complete the proof of \cref{thm:decomposition}, restated below for convenience.

\thmdecomposition*
\begin{proof}
    If $\big||F|-|G|\big|>2k$, we report that $\ted(F,G)>k$.
    If $|F|< k^4$, then we return an empty matching.
    Thus, we henceforth assume that $|F|\ge k^4$ and $\big||F|-|G|\big|\le 2k$.
    We first apply \cref{lem:decomp} to construct a piece decomposition $\DC$ of $F$ into $\Oh(|F|/k^3)$ pieces of length at most $k^3$ each; this takes $\Oh(n)$ time.
    Next, we use \cref{lem:dp} to build a minimum-cost piece matching $\M$ between $\DC$ and $G$; this takes $\Oh(n+|\DC|k^3)=\Oh(n)$ time.
    If the cost of $\M$ exceeds $6k$, then we report that $\ted(F,G)>k$; this is justified by \cref{lem:good_matching_exists}.
    Thus, we henceforth assume that the cost of $\M$ is at most $6k$.
    In particular, there are at most $6k$ unmatched pieces of $\DC$ and at most $6k$ unmatched fragments of $G$.
    By an upper bound on the piece length, there are $\Oh(k^4)$ unmatched characters.
    Moreover, the bound on the matching cost translates to $\Oh(k)$ dirty pieces of $\DC$ with respect to $\M$.
    Consequently, the algorithm of \cref{lem:merge} in $\Oh(n)$ time outputs a matching $\M'$ of size $\Oh(k)$ that leaves the same $\Oh(k^4)$ unmatched characters.
\end{proof}